%% file: synchlics.tex
\documentclass[conference]{IEEEtran}

\usepackage[utf8]{inputenc}
\usepackage{amssymb}
\usepackage{amsthm}
\usepackage{mathtools}
\usepackage{float}
\usepackage{esint}
\usepackage{stmaryrd}
\usepackage{subcaption}
\usepackage{ebproof}
\usepackage{mathtools}
\usepackage[dvipsnames]{xcolor}
\usepackage{tikz}

\usepackage{cleveref}\usepackage{pdfpages}
\usetikzlibrary{calc}
\usetikzlibrary{shapes.geometric}
\usetikzlibrary{shapes.misc}
\usetikzlibrary{shapes.symbols}
\usetikzlibrary{positioning}
\usetikzlibrary{decorations.markings}
\usepackage{scalerel,stackengine}
\setstackgap{S}{1pt}
\usetikzlibrary{external}
\usetikzlibrary{circuits.ee.IEC}
\usepackage{tikzit}

\input{sample.tikzstyles}

\let\oldtikzfig\tikzfig
\renewcommand{\tikzfig}[1]{
	\tikzsetnextfilename{#1}
	\oldtikzfig{#1}
}

\newtheorem{proposition}{Proposition}

\newtheorem{example}{Example}
\newtheorem{lemma}{Lemma}
\newtheorem{theorem}{Theorem}
\newtheorem{definition}{Definition}

\newcounter{storedvalue}

\newcommand{\ground}
{
	\begin{tikzpicture}[circuit ee IEC,yscale=0.9,xscale=0.8]
	\draw[solid,arrows=-] (0,1ex) to (0,0) node[anchor=center,ground,rotate=-90,xshift=.72ex] {};
	\end{tikzpicture}}

\newcommand{\G}{\mathcal{G}} 
\newcommand{\C}{\mathcal{C}} 
\newcommand{\sem}[1]{\llbracket #1 \rrbracket} 

\newcommand{\Stream}{\G^{\omega}} 
\newcommand{\Onion}{\textup{FinApp}(\C)} 
\newcommand{\RegStSeq}[1]{\textup{RegSt}(#1)} 
\newcommand{\Seq}{\textup{Seq}(\C)} 
\newcommand{\fa}[1]{\textup{FA}\!\!\left(#1 \right)} 
\newcommand{\sembis}[1]{\llparenthesis #1 \rrparenthesis} 
\newcommand{\parsem}[2][]{\partial^{#1}#2} 

\newcommand{\D}{\Diamond} 
\newcommand{\bfup}[1]{\textup{\textbf{#1}}} 
\newcommand{\id}{\bfup{id}} 
\newcommand{\Terre}[1]{\ground_{#1}} 
\newcommand{\Dtr}[3]{\bfup{Dtr}_{#1}^{#2}\left(#3\right)} 
\newcommand{\Dtri}[4]{\bfup{Dtr}{[#1]}_{#2}^{#3}\left(#4\right)} 

\newcommand{\UCPM}{\bfup{CPTP}_2} 
\newcommand{\CPM}{\bfup{CPM}_2} 

\newcommand{\Pure}{\textup{PPur}}
\newcommand{\CM}{\textup{(CM)}}
\newcommand{\IM}{\textup{(IM)}}
\newcommand{\AX}{\textup{(Ax)}}

\newcommand{\CC}{\mathbb{C}}

\newcommand{\NN}{\mathbb{N}}

\newcommand{\quot}[2]{{\raisebox{.2em}{$#1$}\left/\raisebox{-.2em}{$#2$}\right.}}
\newcommand{\df}{\stackrel{\scriptscriptstyle\textup{def}}{=}}
\newcommand{\blue}{\textcolor{blue}{\bullet}}
\newcommand{\red}{\textcolor{red}{\bullet}}
\newcommand{\green}{\textcolor{green}{\bullet}}
\newcommand{\orange}{\textcolor{orange}{\bullet}}

\newcommand{\bvdots}{ \tikz[baseline, every node/.style={inner sep=0}]{ \node at (0,0){.}; \node at (0,-6pt){.}; \node at (0,6pt){.}; } }
\newcommand{\ket}[1]{|#1\rangle}
\newcommand{\bra}[1]{\langle#1|}

\DeclareFontFamily{U}{mathb}{\hyphenchar\font45}
\DeclareFontShape{U}{mathb}{m}{n}{
	<5> <6> <7> <8> <9> <10> gen * mathb
	<10.95> mathb10 <12> <14.4> <17.28> <20.74> <24.88> mathb12
}{}
\DeclareSymbolFont{mathb}{U}{mathb}{m}{n}

\DeclareMathSymbol{\pluscirc}{2}{mathb}{"09}

\begin{document}
\title{Graphical Language with Delayed Trace:\\Picturing Quantum Computing with Finite Memory}
\author{\IEEEauthorblockN{Titouan Carette}
\IEEEauthorblockA{Université de Lorraine,\\ CNRS, Inria, LORIA\\F 54000 Nancy, France\\ Email: titouan.carette@loria.fr}
\and
\IEEEauthorblockN{Marc de Visme}
\IEEEauthorblockA{Université de Lorraine,\\ CNRS, Inria, LORIA\\F 54000 Nancy, France\\ Email: marc.de-visme@loria.fr}
\and
\IEEEauthorblockN{Simon Perdrix}
\IEEEauthorblockA{Université de Lorraine,\\ CNRS, Inria, LORIA\\F 54000 Nancy, France\\ Email: simon.perdrix@loria.fr}
}
\maketitle
\begin{abstract}
	Graphical languages, like quantum circuits or ZX-calculus, have been successfully designed to represent (memoryless) quantum computations acting on a finite number of qubits. Meanwhile, delayed traces have been used as a graphical way to represent finite-memory computations on streams, in a classical setting (cartesian data types). We merge those two approaches and describe a general construction that extends any graphical language, equipped with a notion of discarding, to a graphical language of finite memory computations. In order to handle cases like the ZX-calculus, which is complete for post-selected quantum mechanics, we extend the delayed trace formalism beyond the causal case, refining the notion of causality for stream transformers. We design a stream semantics based on stateful morphism sequences and, under some assumptions, show universality and completeness results. Finally, we investigate the links of our framework with previous works on cartesian data types, signal flow graphs, and quantum channels with memories.
\end{abstract}

\IEEEpeerreviewmaketitle

\section{Introduction}
\noindent \textbf{Motivations.} Several graphical languages have been successfully developed for representing finite-dimensional quantum processes. The quantum circuits and the ZX-calculus are the main  examples of such graphical languages. The ZX-calculus is equipped with a complete equational theory \cite{vilmart2018near,carette2019completeness}, that allows, among other applications, 
 to perform circuit optimization \cite{duncan2019graph,kissinger2020reducing}, and to design fault tolerant computations \cite{PhysRevX.10.041030,de2017zx}.
These graphical languages have been designed for finite-dimensional quantum mechanics: each wire represents a finite system -- generally a qubit -- as a consequence, a finite diagram can only represent a finite-dimensional quantum evolution. Notice that using the scalable construction \cite{carette2019szx}, one can represent finite registers, with the possibility to split and merge  registers. This construction makes the representation more compact but it remains a representation of finite-dimensional quantum computations.
There is a fundamental reason for this restriction: finite-dimensional Hilbert spaces, contrary to infinite-dimensional ones, form a compact closed category, and the compact closure is the cornerstone of graphical languages like the ZX-calculus.

To go beyond finite registers, we explore in this paper the design of graphical languages for quantum stream transformations, \emph{i.e.}, computations taking (infinite) sequences of quantum inputs to (infinite) sequences of quantum outputs. Intuitively a transformation acting on a stream of qubits, inputs a qubit and outputs a qubit at each clock tick. In order to allow interactions between systems input at distinct clock ticks, a memory mechanism is required to store some data across the ticks. Such a quantum transformation is called a  \emph{quantum channel with memory} in \cite{kretschmann2005quantum}.

We choose to graphically represent  the memory mechanism using \emph{delayed traces}, \emph{i.e.}, feedback loops that  store qubits from clock tick to the next. The example in Fig. \ref{fig:cnot} consists in applying a CNot gate on consecutive qubits of a stream.

\begin{figure}
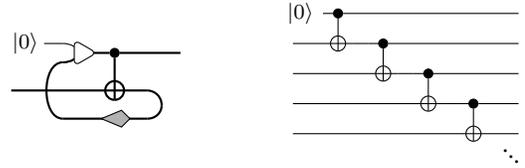

\centerline{$\input{CNot-casc.tikz}\qquad \qquad\input{cnc-intro.tikz}$}
\caption{\emph{  \label{fig:cnot}{\bf Left:} A cascade of CNots on a stream of qubits. At the first tick, a CNot is applied: the control qubit is in the $\ket 0$ state, the target qubit is the first qubit of the input stream. The control qubit is then output and the target qubit stored in the memory. At the second tick the stored qubit becomes the control qubit and the second qubit of the input stream becomes the target qubit and so on. {\bf Right:} An informal unfolded version of the cascade of CNots.}}
\end{figure}

Delayed traces have been studied \cite{sprunger2019differentiable} as a construction which can be applied to any cartesian category. We explore the extension of this construction to the quantum case, since a quantum graphical language, from a category point of view, form a category which is symmetric monoidal but not cartesian.

Depicting finite memory quantum computations on streams does not provide a universal model of quantum computation. It is however an interesting  fragment to explore, strictly more expressive than memory-free languages designed for finite registers, and an intermediate scale model potentially easier to implement using quantum technologies available in the short term, than a universal quantum computer. 

\vspace{0.2cm}

\noindent \textbf{Contributions.}  
We introduce a general construction that extends any (not necessarily quantum) graphical language $\G$ equipped with a discarding map, to a graphical language $\Stream$ with finite memory acting on streams. The construction consists in adding a delayed trace to model the memory, as well as stream constructors and destructors. 

A key property of the construction is that the delay only commutes with causal transformations. Indeed, applying a non-causal transformation before storing a system can produce some side effect on the output at tick $k$ which would occur only at tick $k+1$ if the transformation is applied later on. Moreover, the infinite nature of the computation requires the introduction of a coinduction principle, to show for instance that storing forever a system in a memory and never using it, is equivalent to discard this system right away.

We introduce the \emph{finite approximations} of a $\Stream$-diagram $D$ as a sequence of $\G$-diagrams: the $k^\text{th}$ diagram of the sequence represents the behavior of $D$ from the initial  to the $k^\text{th}$
 tick. The semantics of a $\Stream$-diagram is then defined as the sequence of interpretations of its finite approximations.

 When an evolution is \emph{causal}, 
 its $k^\text{th}$ approximation can be obtained from the $(k+1)^\text{th}$ approximation by discarding the last outputs. This property witnesses the fact that the state of the first $k$  outputs should only depends on the first $k$  inputs. This is however not the case with non-causal evolutions, and in particular post-selected quantum evolutions (\emph{i.e.}, quantum evolutions where one can freely choose the classical outcomes of the measurements). Post-selected evolutions can be represented in several graphical languages including the ZX-calculus. As a consequence, we introduce a new monotonicity condition for finite approximations.
  
  The monotonicity conditions allow us to identify the finite approximations that can be represented in $\Stream$. In particular, these evolutions should satisfy some additional regularity condition, which corresponds to the fact that they can be implemented with a finite memory. 
Finally, we also show the \emph{completeness} of the language, up to some additional assumptions which are satisfied in the quantum case.

\vspace{0.2cm}

\noindent\textbf{Related works.} 
In \cite{sprunger2019differentiable}, a delayed trace construction is introduced for the classical case (cartesian data type), we extend the construction to the non-cartesian case. Moreover we axiomatize, using a graphical language, the delayed trace construction. 
A categorical formulation of delayed traces as infinite combs has also been proposed by \cite{roman2020comb} but again, only in the cartesian or semi-cartesian case.

In categorical quantum mechanics, $!$-boxes \cite{kissinger_et_al:LIPIcs:2015:5533} and the scalable construction \cite{carette2019szx} can be used to represent an infinite family of diagrams at once, and thus a computation acting on a finite but unbounded number of qubits. There is however no notion of stream or memory. 
Since infinite-dimensional Hilbert spaces are not compact closed there are few attempts of graphical languages in infinite dimensions \cite{heunen,coecke2016pictures} --
notice the work of \cite{EPTCS236.4} based on non-standard analysis. Nevertheless, our construction  preserves  compact structures. There is no contradiction here, our scalars are not complex numbers but sequences of complex numbers. Thus, the relevant way to interpret our construction in the categorical quantum mechanic setting would be to consider $\mathbb{C}^\mathbb{N}$-modules instead of infinite-dimensionalal vector spaces over $\mathbb{C}$.

In quantum computing, quantum channels with memory and quantum cellular automata \cite{schumacher2004reversible,arrighi2011unitarity} are examples of computational models with an infinite number of qubits. Notice that typical results in this field are structure theorems which  state for instance that if an evolution is translation invariant and causal then it can be decomposed into a series of local operations. In section \ref{sec:werner}, we discuss the connection between the structure theorem of \cite{kretschmann2005quantum} and the diagrams of $\Stream$.

Delayed traces as been used in \cite{bonchi2014categorical} to axiomatize rational streams. They rely heavily on the properties of linear streams providing a semantics in terms of formal Laurent series. This is the only example we know of previous works on delayed trace in the compact closed, hence non-semi-cartesian, setting. We discuss the links with our formalism in \ref{sec:lin}.

\vspace{0.2cm}

\noindent \textbf{Structure of the paper.} 
We present in Section II some preliminaries on quantum computation, describing how quantum states are represented and the preexisting graphical language of quantum circuits. We then explore the different notions that naturally arise when we add memory to quantum circuits.

In Section III, we rely on the intuitions coming from the quantum case and work at a much greater level of generality: for every graphical language $\G$, we define the graphical language $\Stream$ which manipulates both single inputs and streams of inputs, and allows for the storage of information through time. 

While diagrams of $\Stream$ are finite, we study in Section IV their infinite unfoldings into stateful morphism sequences \cite{sprunger2019differentiable}, and show an equivalence between those sequences that are ultimately constants, and the diagrams of $\Stream$. Those unfoldings are a major intermediate step for the definition of the semantics of diagrams of $\Stream$.

In Section V, we explain how to build a semantics for $\Stream$ from a semantics for $\G$, and show that under some reasonable assumptions, this semantics is \emph{complete}, meaning that the rewriting rules of our language generate all the sound rewriting rules, and \emph{universal}, meaning that the generators of our language generate all the (regular) stream processes.

At last, in Section VI, we explore the applications of the construction, in particular for the ZX-calculus. We also present a fragment of our graphical language, $\Stream_0$, which matches more closely with the preexisting works. In this fragment, one is forced to behave uniformly through time, and operations such that ``changing the third element of a stream'' are not possible.

\vspace{0.2cm}

All the proofs can be found in Appendix \ref{proofs}.

\section{Finite Memory Quantum Computing}

In this section, we review various notions of quantum computing and motivate by examples the kind of computations the language presented in the next section is designed to represent. 

\subsection{Completely Positive Maps}
\label{sec:quantum}
We use the density matrix formalism of finite-dimensionalal quantum mechanics over qubits, see \cite{nielsen2002quantum} for a more complete presentation. We have a symmetric monoidal category $\CPM$ of generalized quantum processes over qubits. The objects are the sets of linear operators of the form $\mathcal{M}_{2^n\times 2^n}\left(\mathbb{C}\right)$ representing systems of $n$ qubits. The morphisms are the linear maps that are completely positive. Among those maps only the trace preserving ones correspond to real physical transformations, we write this subcategory $\UCPM$. The state of an $n$-qubit system is a  \textbf{density matrix} \emph{i.e.} a map $\mathbb{C}\to \mathcal{M}_{2^n\times 2^n}\left(\mathbb{C}\right)$, 
which is positive semi-definite Hermitian and has unit trace. 
Using the Dirac notations $\ket 0 := \scalebox{0.8}{$\begin{pmatrix}
		 1\\
		 0
		 \end{pmatrix}$}$, $\ket 1:= \scalebox{0.8}{$\begin{pmatrix}
		 0\\
		 1
		 \end{pmatrix}$}$, $\bra 0  := \scalebox{1}{$\begin{pmatrix}
		 1&\!\!\!
		 0
		 \end{pmatrix}$}$, and $\bra 1  := \scalebox{1}{$\begin{pmatrix}
		 0&\!\!\!
		 1
		 \end{pmatrix}$}$, the two \emph{classical} states of a qubit are the density matrices  $\ket0 \bra 0$ and $\ket1 \bra 1$.

The monoidal product is the tensor product of vector spaces $\mathcal{M}_{2^n\times 2^n}\left(\mathbb{C}\right)\otimes\mathcal{M}_{2^m\times 2^m}\left(\mathbb{C}\right)\simeq\mathcal{M}_{2^{n+m}\times 2^{n+m}}\left(\mathbb{C}\right)$. The symmetry maps are the exchange maps $\rho \otimes \nu \mapsto \nu \otimes \rho$. We denote $f^\dagger$ the Hermitian adjoint of a completely positive map. A map is said to be an \textbf{isometry} if $f^\dagger \circ f=id$. 
A quantum process is said to be $\textbf{pure}$ if it is of the form $\rho \mapsto V\rho V^\dagger$ with $V\in \mathcal{M}_{2^{m}\times 2^{n}}\left(\mathbb{C}\right)$. A \textbf{measurement} which maps  $\rho$ to $\ket{0}\bra{0}\rho\ket{0}\bra{0}+\ket{1}\bra{1}\rho\ket{1}\bra{1}$, is an example of non-pure evolution. 
In $\CPM$, we can also represent \textbf{post-selected} measurements, those are non-physical processes where we assert that a measurement gave a chosen outcome. For instance, choosing the  outcome $\ket{0}\bra{0}$   corresponds to the post-selected measurement
$\rho \mapsto \bra{0}\rho\ket{0}$.

Let the \textbf{discard} map be $\ground\df \rho \mapsto Tr(\rho)$, which corresponds to measuring a qubit and forgetting the result. A quantum evolution $f$ is \textbf{causal} if $\ground \circ f= \ground$. Intuitively, causal evolutions are side-effect free. 

This discard map is not pure and can even be seen as the essence of all impurity in the following sense: for any completely positive map $f:A\to B$ there is a system $C$ and a pure map $p:A\to B\otimes C$ such that $f=\left(id_B\otimes\ground_C\right) \circ p$. In this situation $p$ is said to be a \textbf{purification} of $f$. Purifications are not unique, but they are in fact unique up to isometries:
 
\begin{theorem}[Stinespring dilation \cite{stinespring1955positive}]\label{thm:stinespring}
	Given two purifications $p:A\to B\otimes C$ and $p':A\to B\otimes C'$ of the same completely positive map $f:A\to B$, either there is an isometry $v:C\to C'$ such that $p'=\left(id_A \otimes v \right)\circ p$, or there is an isometry $v':C'\to C$ such that $p=\left(id_A \otimes v' \right)\circ p'$.
\end{theorem}

\subsection{Quantum Gates}

We represent the maps of $\UCPM$ as gates in circuits. This is an example of graphical language that will be formally defined in the next section. The composition corresponds to plugging gates and the tensor product to putting them side by side. Note that usually quantum circuits cannot represent $\CPM$ in full generality. But other graphical languages like the ZX-calculus have been designed for this. We only present a few gates, taken both from the quantum circuits and ZX-calculus, that we will use in examples.

\begingroup
\setlength{\tabcolsep}{15pt} 
 \[\begin{array}{rclcrcl}
		\input{ket0.tikz}&= & \ket{0}\bra{0} &\qquad  & \input{terre.tikz}&=& \rho \mapsto Tr(\rho)\\[0.3cm]
		\input{mix-.tikz}&=&\frac{1}{2}\scalebox{0.8}{$\begin{pmatrix}
		 1&\!\!0\\
		 0&\!\!1
		 \end{pmatrix}$}&&\input{bra0.tikz}&=&\rho \mapsto \bra{0}\rho\ket{0}
	\end{array}\]
\endgroup

The first gate takes no input and produces the pure state $\ket{0}\bra{0}$. The second is the discard map and the third produces the \textbf{maximally mixed} 1-qubit state. The fourth gate is the post-selected measurement selecting the outcome $\ket{0}\bra{0}$. We also have gates acting on more than one qubit:

\[\input{bellp-2.tikz} ~=~ \frac{1}{2}\scalebox{0.8}{$\begin{pmatrix}
			1&0&0&1\\
			0&0&0&0\\
			0&0&0&0\\
			1&0&0&1
			\end{pmatrix}$}\qquad\quad \input{cswap.tikz} ~=~ \rho \otimes \nu  \mapsto \nu \otimes \rho \]
			
			\[\input{cnot.tikz}~=~\rho\mapsto \scalebox{0.8}{$\begin{pmatrix}
		1&0&0&0\\
		0&1&0&0\\
		0&0&0&1\\
		0&0&1&0
		\end{pmatrix}$}\rho \scalebox{0.8}{$\begin{pmatrix}
		1&0&0&0\\
		0&1&0&0\\
		0&0&0&1\\
		0&0&1&0
		\end{pmatrix}$}\]

The first state is a \textbf{Bell pair}. This state is  entangled, meaning it cannot be written as the tensor product of two one-qubit states. We can even say that it is the maximally entangled state in the sense that discarding one qubit of the pair turns the other into the maximally mixed state. The swap exchanges two qubits and the CNot gate is a pure map acting as $\ket{x}\otimes \ket{y}\mapsto\ket{x}\otimes \ket{x\oplus y}$ on the computational basis. We are now ready to provide concrete examples of quantum computation with memory.  

\subsection{Quantum Computation with Memory}

In order to go beyond quantum computation acting on a finite register of qubits, it is natural to consider streams of qubits: we consider a global clock, such that at each clock tick some qubits are input. For allowing interactions across clock ticks, like applying a CNot on two qubits input at distinct clock ticks, a memory mechanics is required to store a qubit and intuitively wait for another qubit to be available. Quantum channels with memory, introduced by Kretschmann and Werner \cite{kretschmann2005quantum},  can be informally depicted as follows:\footnote{In \cite{kretschmann2005quantum}, the authors mainly consider the case of a clock without initialization, \emph{i.e.}, clock ticks in $\mathbb Z$ rather than $\NN_{\geq 1}$}

~\\

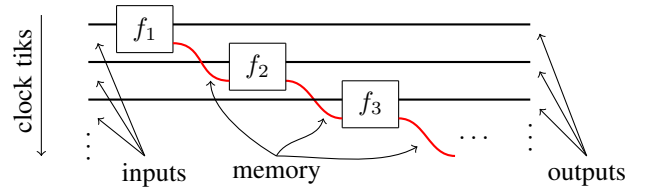
\input{QCwithMem.tikz}

~\\

Thus the behavior of the computer at clock tick $k>0$ is a quantum process $f_k: A_k \otimes M_{k-1} \to B_k \otimes M_k$, with $M_0\df \mathbb{C}$. Following the terminology for such processes in the  classical case \cite{sprunger2019differentiable}, we call such collection of processes 
a \textbf{stateful morphism sequence}. We give the example of a cascade of CNots gates (see Fig. \ref{fig:cnot}). At first clock tick the memory is initialized with $\ket{0}\bra{0}$. At each subsequent clock tick, a CNot is applied, the control qubit being the memory qubit and the target being the input qubit. Finally, the memory qubit is output and the input qubit is stored in the memory. 
The corresponding stateful morphism sequence is:

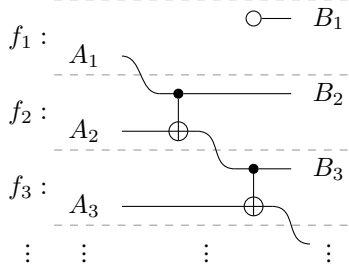
\begin{center}	
	\input{cnc.tikz}
\end{center}

In practice, one cannot access the whole infinite computation at once, but only what has been computed up to some clock tick $k$. To stop the computation of a stateful morphism at clock tick $k$, we discard the memory system $M_k$ and obtain, by plugging the memories, a process $\bigotimes\limits_{i=1}^{k} A_i \to \bigotimes\limits_{i=1}^{k} B_i$ called the \textbf{finite approximation} at clock tick $k$. For the cascade of CNot the sequence of finite approximations is:

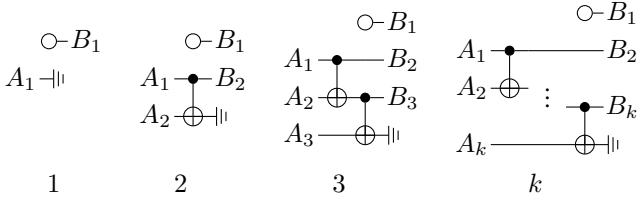
\vspace{0.2cm}

\centerline{
	\begin{tabular}{cccc}
		\input{cnc2.tikz}&\hspace{-0.3cm}\input{cnc3.tikz}&\hspace{-0.3cm}\input{cnc4.tikz}&\hspace{-0.3cm}\input{cnc5.tikz}\\[1cm]
		$1$&\hspace{-0.5cm}$2$&\hspace{-0.5cm}$3$&\hspace{-0.5cm}$k$
	\end{tabular}
}

\vspace{0.2cm}

A stateful morphism sequence leads to a unique sequence of finite approximations. However, two different stateful morphism sequences can have the same sequence of finite approximations, they are then said to be \textbf{observationally equivalent}. In Section IV, we characterize  the observationally equivalent stateful morphism sequences.

\subsection{Finite Approximations and Causality}\label{sec:FA}
\label{sec:monotone}
Another important question is to characterize the sequence of finite approximations that can be produced by a stateful morphism sequence. A first guess is that there are exactly the sequences for which the behavior at clock tick $k$ does not depend on what happens at the clock ticks $k'>k$. In other words, the present only depends on the past and not on the future. 
More formally, given a sequence of finite approximations $(f_k)_{k>0}$ with $f_k:A_1\otimes \ldots \otimes A_k\to B_1\otimes \ldots \otimes B_k$, the condition is, for any $k>0$, 

\[\input{monotone_q.tikz} \quad = \quad \input{monotone_q2.tikz}\]

This condition is called \emph{causality} in the classical setting \cite{sprunger2019differentiable}, and \emph{one-way signaling} in the context of categorical quantum mechanics \cite{kissinger2017categorical} since causality has a different meaning (see section \ref{sec:discard}). This condition is also at the heart of the quantum channels with memory in \cite{kretschmann2005quantum}. 

This one-way signaling condition characterizes the sequences of finite approximations produced by sequences of \emph{causal} stateful morphisms. In particular, this notion is well adapted to $\UCPM$ where all morphisms are  causal. However, we aim at considering sequences of non-causal  stateful morphisms, like in $\CPM$. As a consequence, we need to introduce a weaker monotonicity condition to cover the non-causal case.   Useful examples of non-causal evolutions are the post-selected quantum evolutions. 
 
 In post-selected quantum mechanics the present can  depend on the future in a very specific way. It might happen that waiting gives us more information on a given state, for example turning a mixed state into a pure one. An example is given by the following post-selection protocol (see Example \ref{ex:Bell} in Section \ref{sec:delayed} for a pictorial description using a delayed trace): at each clock tick, a post-selected measurement of the memory is performed (except at the very first clock tick), then a new Bell pair is produced. One of the qubit of the pair is directly output while the other one is stored in the memory.
 The corresponding stateful morphism sequence is:

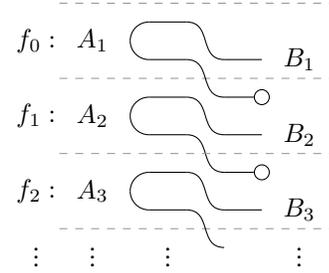
\begin{center}
	\input{bellb.tikz}
\end{center}

Notice that we have  $\input{bell-meas-1.tikz}=\input{bell-meas.tikz}$, indeed the post-selected measurement is actually implemented by a linear map, hence the scalar $\input{scal-.tikz}=1/2$ which is witnessing the fact that an actual measurement produces this particular outcome with probability $1/2$. 

The finite approximations of this protocol are:

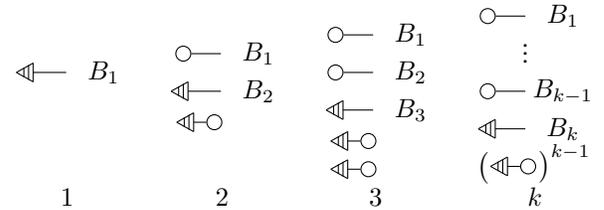
\begin{center}
	\begin{tabular}{cccc}
		\input{bell1-.tikz}&\input{bell2-.tikz}&\input{bell3-.tikz}&\input{bell4-.tikz}\\
		$1$&$2$&$3$&$k$
	\end{tabular}
\end{center}

We see 
that if we stop at tick $k$, we have no information on the $k^\text{th}$ output, which is the maximally mixed state. However at tick $k+1$, the post-selection will force the $k^\text{th}$ output to be $\ket{0}\bra{0}$. In a sense, the future can refine the present. To formalize this we use the \textbf{Loewner order} defined on density matrices as $\rho \sqsubseteq \nu$ if $\nu - \rho$ is positive semi-definite. It can be extended naturally to maps by $f\sqsubseteq g$ if $g-f$ is completely positive. The Loewner order characterizes what it means for a state to be more precise than another. 
Since the morphisms of $\CPM$ can also be trace increasing, we consider a lax-version of the Loewner order: 
$f\preceq g$ if $\exists \lambda > 0$ such that $f\sqsubseteq \lambda g$. We will show in Section \ref{sec:sem} that the \textbf{monotone} sequences of finite approximations, \emph{i.e.}, such that $\forall k>0$, 
\[\input{monotone_q.tikz} \quad  \preceq \quad \input{monotone_q2.tikz}\]
 are exactly the ones that approximate the stateful morphism sequences in $\CPM$.

\subsection{Regularity and Finite Memory}

Among all the possible stateful morphism sequences, we  focus on the regular ones, \emph{i.e.}, those that are eventually constant: a stateful morphism sequence $(f_k)_{k>0}$ is regular if $\exists n, \forall k>n, f_k=f_n$. Notice in particular that regular stateful morphism sequences use a bounded amount of memory. Hence, it represents quantum computations acting on an unbound number of inputs but with a finite memory. This model is not a universal model of quantum computation, but it is an interesting  fragment to explore, and an intermediate scale model potentially easier to implement than a universal quantum computer using technologies which will be available in a near future. 

Intuitively,  regular  stateful morphism sequences can be finitely described as they are eventually constant. In the next section, we introduce a graphical construction which turns any graphical language equipped with a discard, and acting on finite registers, into a graphical language for regular   stateful morphism sequences, \emph{i.e.}, in the quantum case, a language for representing and reasoning about finite-memory quantum stream computation.  

\section{The Colored Graphical Language $\Stream$}

In this section we define the syntax of our language. We first set the string diagram notation and then, starting with a graphical language $\G$ describing computations we define a colored prop $\Stream$ which depicts computation with memory.

\subsection{String Diagrams}

We use the framework of props. A \textbf{prop} is a small symmetric strict monoidal category whose monoid of objects is freely spanned by one element denoted $1$. Representing the tensor additively, every object is then of the form $1+\cdots+1$ and therefore denoted $n$. The unit of the tensor is $0$. The categories $\UCPM$ and $\CPM$ from the previous section are equivalent to props\footnote{$\UCPM$ and $\CPM$ can be turned into props by defining the objects to be natural numbers and the morphisms $n\to m$ to be the linear maps $\mathcal M_{2^n\times 2^n}(\mathbb C) \to\mathcal M_{2^m\times 2^m}(\mathbb C)$.}. A \textbf{colored prop} is a small symmetric strict monoidal category whose monoid of objects is freely spanned by a set $C$ of colors. Each object can then be written as a list $c_1 + \cdots +c_k$ of colors. The string diagram notation represents arrows as boxes with colored wires as inputs and outputs. Given a $\{\blue,\red,\green,\orange\}$-colored prop $\textbf{P}$ with the arrows $F:\blue+\red \to \green+\orange$, $G:\blue\to\red$, $H:\red\to \green$ and the swap map $\sigma_{\blue,\red}: \blue+\red \to \red +\blue$ we write:

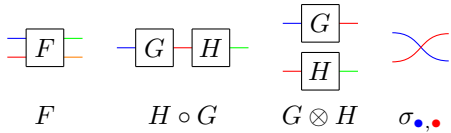
\begin{center}
	\begin{tabular}{cccc}
		$\input{ord0.tikz}$ & $\input{comp.tikz}$ & $\input{tens.tikz}$ & $\input{swap.tikz}$\\[0.5cm]
		$F$ & $H\circ G$ & $G\otimes H$ & $\sigma_{\blue,\red}$
	\end{tabular}
\end{center}

We say that a colored prop has a \textbf{compact structure} when for each color $\red$ there are two arrows $\input{cup.tikz} :0\to \red+\red$ and $\input{cap.tikz}: \red+\red \to 0$, satisfying the following equations:

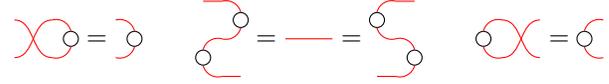
\begin{center}
	$\input{capsym0.tikz}=\input{cap.tikz}\qquad\input{snake0.tikz}=\input{snake1.tikz}=\input{snake2.tikz}\qquad\input{cupsym0.tikz}=\input{cup.tikz}$
\end{center}

The monochromatic prop $\CPM$ admits a compact structure but not $\UCPM$. Referring to string diagrams we will often say \textbf{graphical language} for a monochromatic prop and \textbf{colored graphical language} for a colored prop. We write $\Gamma \vdash D=K$ when we can rewrite the diagram $D$ into the diagram $K$ using the rewriting rules $\Gamma$.

To define the finite approximations, we need a way to express the loss of the data stored in the memory when we stop the computation, in other words we need a discard map.

\begin{definition}[Discard]
	A \textbf{discard} prop is a prop $\textbf{P}$ where we fix a \textbf{discard} map $\Terre{1}:1\to 0$. We denote $\Terre{0}\df id_0$ and $\Terre{a+b}\df \Terre{a}\otimes\Terre{b}$.
\end{definition}

In a discard prop a morphism $C:a\to b$ is said to be \textbf{causal} if: $\input{cdist0.tikz}= \input{cdist1.tikz}$.

In what follows we consider only monochromatic discard props. $\UCPM$ and $\CPM$ are both discard. In $\UCPM$ all maps are causal while in $\CPM$ the causal maps are exactly the ones from $\UCPM$. Let $\G$ be a monochromatic discard graphical language defined by generators and equations. We build a colored graphical language $\Stream$ representing stream transformers.

\subsection{Type System}

The colors of the colored prop $\Stream$ are given by:

\begin{center}
	$C \df ~1~ | ~\omega~ | ~\D C$
\end{center}

A way to interpret the types is to consider a global clock that starts at the beginning of the computation. The $1$ type represents the basic data processed by $\G$. We keep in $\Stream$ all the generators and equations of $\G$. This yields an inclusion functor $\iota:\G\to \Stream$ that we will keep implicit most of the time. A wire of type $1$ sends one unit of data at tick $0$ and nothing after. The tensor unit is denoted $0$. 

The $\omega$ type represents a stream of basic data. A wire of type $\omega$ sends one unit of data at each tick. We write $n\omega$ for a stream of $n$-tuples of data, \emph{i.e.}, $\omega+\ldots +\omega$ $n$-times, with the convention $1\omega\df \omega$ and $0\omega\df 0$. For each generator $G:n\to m$ in $\G$ we define a generator $\omega g: n\omega \to m\omega$ in $\Stream$. This gives a functor $\omega: \G \to \Stream$.

The \textbf{delay modality} $\D$ can be applied to any color to produce a delayed color. We write $\D^n C$ for the color $C$ delayed $n$ times: $\D^0 C = C$, $\D^n 0\df 0$ and $\D^{n+1} C = \D(\D^n C)$. A wire of type $\D^n 1$ sends one unit of data at tick $n$ but nothing before or after that tick. A wire of type $\D^n \omega$ sends nothing until tick $n+1$ and then sends one unit of data at each tick. The delay modality is extended to tensors of colors by setting $\D (S+T)=\D S + \D T$. For each generator $g:a\to b$ we have a delayed generator $\D g: \D a \to \D b$. This then extends to an endofunctor $\D: \Stream \to \Stream$.

We see that $\G$ is a subcategory of $\Stream$ in various ways given by the functors $\D^n \circ \iota :\G\to \Stream$ and $\D^n \circ \omega :\G\to \Stream$.

As an example of how the type system works, a finite stream of size $3$ sending two units of data at each tick for the first $3$  ticks would have type: $1+1+\D 1 + \D 1 + \D^2 1 + \D^2 1= 2+ \D 2 + \D^2 2$. Here we see that the tensor is used to encode at the same time spatial and temporal juxtaposition. 

\subsection{Initialization and Derivative}

To manipulate streams we add to $\Stream$ two dual operators:

\begin{center}
	\begin{tabular}{cc}
		$\input{diff.tikz}:\omega \to 1+\D\omega\quad$&$\quad\input{init.tikz}:1+\D\omega \to \omega$\\[0.5cm]
		\textbf{stream derivative} & \textbf{stream initialization}
	\end{tabular}
\end{center}

The derivative decomposes a stream into one data at first tick and a delayed stream which corresponds to the usual stream derivative. The initialization takes a delayed stream and adds a bit of data at the beginning to make it undelayed. They interact according to:

\begin{center}
	$\input{elim0.tikz}\stackrel{\left(\triangleright\triangleleft\right)}{=}\input{elim1.tikz}$ and $\input{exp0.tikz}\stackrel{\left(\triangleleft\triangleright\right)}{=}\input{exp1.tikz}$.
\end{center}

They also satisfy a distribution rule with all the delayed omega generators:

\begin{center}
	$\input{dist0.tikz}\quad\stackrel{\left(\blacktriangleright\right)}{=}\quad\input{dist1.tikz}$
\end{center}

The dual equation $\left(\blacktriangleleft\right)$ for derivatives is also true and follows from $\left(\triangleleft\triangleright\right)$, $\left(\triangleright\triangleleft\right)$ and $\left(\blacktriangleright\right)$. Like all generators in $\Stream$, initialization and derivative admit delayed versions of types $\D^n\omega \to \D^n 1+\D^{n+1}\omega$ and $\D^n 1+\D^{n+1}\omega \to \D^n \omega$ satisfying the same equations. Those generators and equations are similar to the ones used in the scalable notations of \cite{carette2019szx}. There, such triangles were used to represent finite spatial juxtaposition of data while our generators deal with infinite temporal juxtaposition. 

\begin{example}\label{ex:scal} Example of quantum circuit stream diagram with no input/output. This diagram has side effects. At each clock tick the scalar $\input{scal-.tikz}= 1/2$ is produced. 
\[\input{exa1.tikz}~\stackrel{\left(\triangleleft\triangleright\right)}{=}~\input{exa2.tikz}~\stackrel{\left(\blacktriangleleft\right)}{=}~\input{exa3.tikz}~\stackrel{\left(\blacktriangleright\right)}{=}~\input{exa4.tikz}\]

\end{example}

\subsection{Delayed Trace}\label{sec:delayed}

An important ingredient to the construction is the \textbf{delayed trace}. It is not strictly speaking a new generator but a constructor similar to the trace in traced monoidal categories. Given a map $D:a+ \D c \to b+c$ we can trace it to construct a new map $\Dtr{c}{a,b}{D}: a\to b$ represented by:

\begin{center}
	$\input{trace.tikz}$
\end{center}

It allows a process to take as an input at tick $k+1$ one of its own outputs at tick $k$. In other words it allows to represent the memories of a stateful morphism sequence. It satisfies the following trace-like axioms:

\begin{center}
	\begin{tabular}{cc}
		$\input{tracenat0.tikz}=\input{tracenat1.tikz}$ & $\input{tracenat2.tikz}=\input{tracenat3.tikz}$\\[0.75cm]
		
		$\input{tracenat4.tikz}=\input{tracenat5.tikz}$ & $\input{tracenat6.tikz}=\input{tracenat7.tikz}$
		
	\end{tabular}
\end{center}

The axiom for the tensor unit is here a tautology: $\Dtr{0}{a,b}{D}=D$. For typing reason we cannot obtain the identity if we trace the swap. Instead we get the \textbf{delay}: $\input{delay1.tikz}: c\to \D c$.

\begin{center}
	$\input{delay1.tikz}\df \input{delay0.tikz}$
\end{center}

The delay holds its input at tick $k$ and releases it at tick $k+1$. When $\G$ has a symmetric compact structure we can recover the delayed trace from the delay:

\begin{center}
	\scalebox{0.7}{$\input{comptrace0.tikz}=\input{comptrace5.tikz}=\input{comptrace1.tikz}=\input{comptrace2.tikz}=\input{comptrace3.tikz}=\input{comptrace4.tikz}$}
\end{center}

\begin{example}\label{ex:Bell} The protocol described in Section \ref{sec:FA} can be depicted as follows.  At each clock tick, a post-selected measurement of the memory is performed and a new Bell pair is produced. One of the qubit of the pair is directly output while the other one is stored in the memory. 
\[\input{Bell-stream.tikz}=\input{Bell-stream-2.tikz}=\input{Bell-stream-3.tikz}\]
\end{example}

\subsection{Causal Maps}

Another axiom of the trace that we do not have is the analog of dinaturality. In fact this would imply that everything commutes with the delay:

\begin{center}
	\scalebox{0.8}{$\input{kom_3.tikz}=\input{kom_2.tikz} \Rightarrow \input{kom_1.tikz}=\input{kom_0.tikz}$}
\end{center}

Intuitively, it does not make any difference if we apply $K$ on the memory data 
at tick $k$ or at tick $k+1$. We want this to be true only if $K$ has no side effect, otherwise, we would identify processes that are not observationally equivalent, \emph{e.g.}, a process that diverges at tick 3 with a process that diverges at tick 4. Thus we require this from swaps, derivatives and initializations:

\begin{center}
	\scalebox{0.8}{$\input{delaydist0.tikz}\stackrel{\left(\triangleleft\right)}{=} \input{delaydist1.tikz}\qquad \input{delaydist2.tikz}\stackrel{\left(\triangleright\right)}{=}\input{delaydist3.tikz}$}
\end{center}
\begin{center}
	\scalebox{0.8}{$\input{sidist0.tikz}\stackrel{\left(\sigma\right)}{=} \input{sidist1.tikz}$}
\end{center}

and from all causal maps $C$ of $\G$:

\begin{center}

	\scalebox{0.8}{$\input{cdist0.tikz}= \input{cdist1.tikz}\quad \Rightarrow\quad \input{dcom2.tikz}=\input{dcom3.tikz}$} $\quad\left(\ground\right)$.\\

\end{center}

Note that the swaps coming from $\G$ are causal by definition and thus $(\sigma)$ follows from $(\ground)$. However, some swaps like the one between $1$ and $\omega$ do not come from morphisms of $\G$, hence the need for the axiom $(\sigma)$. Using the upcoming rule $(\D_{\mathbb{N}})$, we are able to derive the same equation with $\D^n \omega C$ instead of $\D^n \iota C$.

\begin{example}The rules defined so far can be used to unfold the first step of the cascade of CNots: 
\[\input{CNot-casc-1.tikz} \stackrel{\left(\triangleleft\triangleright\right)}{=} \input{CNot-casc-2.tikz} \stackrel{\left(\blacktriangleright\right)}{=} \input{CNot-casc-3.tikz}\]\[\stackrel{\UCPM}{=} \input{CNot-casc-4.tikz}\stackrel{\left(\triangleright\right)}{=}\input{CNot-casc-5.tikz}\]
\end{example}

\subsection{Idempotents}

Another identification we want to make with respect to observational equivalence concerns idempotents. Indeed, if an idempotent is applied on the memory data at tick $k$ then applying it again at tick $k+1$ has no effect at all. So it is possible to send a copy of an idempotent through the delay.

\begin{center}
	\scalebox{0.75}{$\input{pidist0.tikz}= \input{pidist1.tikz} \Rightarrow \input{pdcom2.tikz}=\input{pdcom3.tikz}$} $\quad\left(\pi\right)$\\
\end{center}

Using the upcoming rule $(\D_{\mathbb{N}})$, we are able to derive the same equation with $\D^n \omega \pi$ instead of $\D^n \iota \pi$. We note that the axiom $(\pi)$ implies:

\begin{center}
	 $\input{pidist2.tikz}=\input{pidist3.tikz}$.
\end{center}

Causal and idempotent morphisms are not the only two kind of morphisms that interact with the delay in some way. However, because of our completeness result (\cref{thm:completeness}), we know that we can derive all those interactions from the present axioms.

\subsection{Coinduction}

We now introduce the final elements of the streamed prop construction. We denote by $\AX$ all the axioms that have been presented so far: $\left(\triangleright\right)$, $\left(\triangleleft\right)$, $\left(\blacktriangleright\right)$, $\left(\blacktriangleleft\right)$, $\left(\triangleright\triangleleft\right)$, $\left(\triangleleft\triangleright\right)$, $\left(\sigma\right)$, $\left({\ground}\right)$, $\left(\pi\right)$, and the four trace-like axioms.

 To obtain the final axiomatization of $\Stream$ we quotient by a coinduction meta rule. Given two sequences of diagrams $\left(S_i\right)_{i\in\mathbb{N}}$ and $\left(T_j\right)_{j\in\mathbb{N}}$ in $\Stream$:
 
\begin{center}
	
	$\begin{prooftree}
		\hypo{\forall n\in \mathbb{N}\qquad \AX, \left[\D S_{n+1}= \D T_{n+1}\right] \vdash S_n =T_n}
		\infer1[$\left(\D_\mathbb{N}\right)$]{\AX  \vdash S_0=T_0}
	\end{prooftree}$
	
\end{center}

In practice we will use a weaker form which corresponds to the constant case where for all $i,j\in \mathbb{N}$, $S_i=S$ and $T_j =T$:

\begin{center}
	
	$\begin{prooftree}
	\hypo{\AX, \left[\D S= \D T\right] \vdash S =T}
	\infer1[$\left(\D\right)$]{\AX  \vdash S=T}
	\end{prooftree}$
	
\end{center}

Note that the weak form is also equivalent to the strong one in the case of eventually constant sequences of diagrams.
We provide an example of the coinduction principle in action. 

\begin{example}
	We can show that discarding is the same as storing forever in a memory:
	
	\begin{center}
		$\input{forever0.tikz}=\input{forever1.tikz}$
	\end{center}

	By applying $\left(\D\right)$ to show the equality, we can assume that its delayed version holds.
	
	\begin{center}
		$\input{forever0.tikz}\stackrel{\left(\blacktriangleright\right)}{=}\input{forever2.tikz}\stackrel{\left(\D\right)}{=}\input{forever4.tikz}\stackrel{\left({\ground}\right)}{=}\input{forever1.tikz}$
	\end{center}
	Note that we use the fact that the discard map is causal in the last step.
\end{example}

 With a similar proof we can also derive the rules $\left({\ground}\right)$ and $\left(\pi\right)$ for the $\omega$ generators. Our definition of $\Stream$ is now complete.

\section{$\Stream$ and Stateful Morphism Sequences}
 
We aim to describe every diagram $D\in\Stream$ with stateful morphism sequences, \emph{i.e.}, a sequence of diagrams of $\G$, the k$^\text{th}$ diagram of the sequence representing the behavior of $D$ at clock tick $k$. 

\subsection{Stateful Morphism Sequences}

We start by defining the stateful morphism sequences over $\G$.

\begin{definition}
	A \textbf{stateful morphism sequence} $f$ over a prop $\G$ is given by three sequences of objects $\left(a_i\right)_{1\leq i}$, $\left(b_i\right)_{1\leq i}$ and $\left(m_i\right)_{0\leq i}$ with $m_0\df 0$ together with a sequence of maps $f_i: a_i \otimes m_{i-1} \to b_i\otimes m_i$ with $1\leq i$. $f_i$ is called the $i$th \textbf{layer} of $f$.
\end{definition}

{A {stateful morphism sequence} $f$ can be (informally) depicted as follows: Layers are separated by dash lines and we connect in red the memory of consecutive layers.}
{\[\input{gmap0.tikz}\]}

Following \cite{sprunger2019differentiable} we define a category $\textup{St}(\G)$ of stateful morphism sequences over $\G$. The objects are the sequences $\left(a_i\right)_{1\leq i}$ of objects in $\G$. When the types match, $g\circ f$ and $f \otimes g$ are defined as:
\begin{center}
	$(g\circ f)_k= \input{compseq-.tikz}\qquad(f\otimes  g)_k= \input{comptens-.tikz}$
\end{center}

We define a delay operator on stateful sequence morphism as $\left(\D f\right)_1 \df id_0$ and $\left(\D f\right)_k \df f_{k-1}$.

Clearly, all diagrams in $\Stream$ being finite we cannot represent arbitrary sequences. In fact we will see that we can only represent   $\RegStSeq{\G}$, the subcategory of $\textup{St}\left(\G\right)$  restricted to  \textbf{regular} stateful morphism sequences in which we consider only eventually constant sequences.

\subsection{Finite Approximations}

	A \textbf{finite aproximation sequence} over a discard prop $\G$ is given by two sequences of objects $\left(a_i\right)_{1\leq i}$ and $\left(b_i\right)_{1\leq i}$, together with a sequence of maps $f_k: \bigotimes\limits_{i=1}^k a_i \to \bigotimes\limits_{i=1}^k b_i$. $f_k$ is called the $k^\text{th}$ \textbf{approximation} of $f$. We postpone to Definition \ref{def:fa} the precise conditions that those sequences must satisfy.
Given a stateful morphism sequence we define its finite approximation sequence $\textup{FA}(f)$ by:

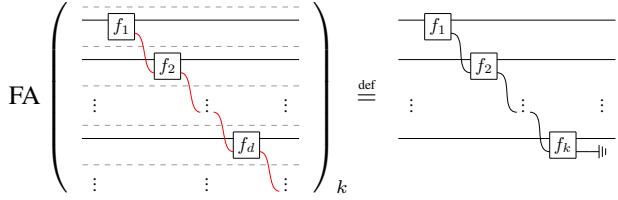
\begin{center}
	$\textup{FA}\left(\scalebox{0.7}{\input{gmap0.tikz}}\right)_k\df\scalebox{0.7}{\input{fadef.tikz}}$
\end{center}

There is also a category $\textup{FA}(\G)$ of finite approximation sequences over $\G$ where composition and tensor are defined approximation-wise. We have a symmetric monoidal functor $\textup{FA}:\textup{St}(\G) \to \textup{FA}(\G)$. Two stateful sequences with the same image by $\textup{FA}$ are said to be observationally equivalent. We design $\Stream$ towards the goal to be universal and complete for $\textup{FA}(\G)$. we will see that while we do not achieve this goal in full generality, we can get close enough in some particular cases.

\subsection{Stratified Types}

There are more types in $\Stream$ than in $\RegStSeq{\G}$. We need to quotient in order to obtain a correspondence. A type in $\Stream$ is a list of colors of two kinds: $\D^k 1$ and $\D^k \omega$. We call the \textbf{degree} of a type the highest delay appearing in it. A type $a$ of degree $d$ is said to be \textbf{stratified} if it is of the form: $a=\sum_{i=0}^{d} \D^i n_i + \D^{d} n_{d+1}\omega$, i.e., colors are sorted by increasing delays with at the end the $\omega$ colors, all having the largest delay. Note that some $n_i$ might be $0$. Given a stratified type $a$ of degree $d$ we define the \textbf{rising} $\delta^k a$ with $k\geq d$ as:

\begin{center}
	 $\delta^k a\df \sum_{i=0}^{d} \D^i n_i + \sum_{j=d}^{k} \D^{j} n_{d+1} + \D^{k} n_{d+1}\omega$.
\end{center}

Note that $\delta^d a = a$. $\delta^k a$ is a stratified type of degree $k$ if $a$ contains $\omega$ types and $d$ otherwise. We say that two types are \textbf{disjoint} if there is no tick when they both send data. Formally, we inductively define a disjointness symmetric relation $*$ on types as: $0*a\Leftrightarrow a\neq 0$, $\D^i 1 * \D^j 1 \Leftrightarrow i\neq j$ ,$\D^i 1 * \D^j \omega \Leftrightarrow i< j$ and $a*(b+c) \Leftrightarrow \left(a*b\right) \land \left(a*c\right)$. A \textbf{disjoint swap} is a swap $a+b\to b+a$ between disjoint types, i.e., such that $a*b$. 

The category of \textbf{disjoint maps} is defined as the wide subcategory of $\Stream$ spanned by initializations, derivatives and disjoint swaps. Note that the subcategory of disjoint maps is a groupoid since all generators are invertible and their inverses are also disjoint maps. The disjoint maps do not modify anything about the processing of data but are graphical bureaucracy that we need to quotient.

\begin{lemma}
	Given a type $a$ of degree $d$. There is always a disjoint map $a\to \overline{a}$, where $\overline{a}$ is stratified of degree $d$, as well as a disjoint map $\overline{a}\to \delta^k\overline{a}$ for each $k\geq d$. And if there is a disjoint map $a\to b$ then it is unique and there is a $k$ such that $\delta^k \overline{a}=\delta^k\overline{b}$.
\end{lemma}

We define the equivalence relation $\sim_s$ on types as $a\sim_s b$ whenever there exists a (necessarily unique) disjoint map $a\to b$. So if $a\sim_s c$ and $b\sim_s d$ then there is a natural bijection between $\Stream(a,b)$ and $\Stream(c,d)$ whose components are the disjoint maps $a\to c$ and $b\to d$. This implies that we have a well defined monoidal category $\quot{\Stream}{\sim_s}$. We write $D\sim_s D'$ whenever two diagrams are identified by this quotient.

For any type $a$ of degree $d$ the equivalence class of $a$ for $\sim_s$ is simply the set $\{b~|~ \exists k, \delta^k\overline{b} = \delta^k \overline{a} \}$. Those equivalence classes are in bijection with the ultimately constant sequences of types in $\G$. Thus, the objects of $\quot{\Stream}{\sim_s}$ and $\RegStSeq{\G}$ are in bijection.

We define the functor $G:\RegStSeq{\G}\to \quot{\Stream}{\sim_s}$ from the functor $\RegStSeq{\G}\to \Stream$ by mapping the sequence of objects in $\G$ eventually constant from tick $d$ to the corresponding stratified type of degree $d$ in $\Stream$. A stateful morphism sequence is then mapped to a diagram as follows: 

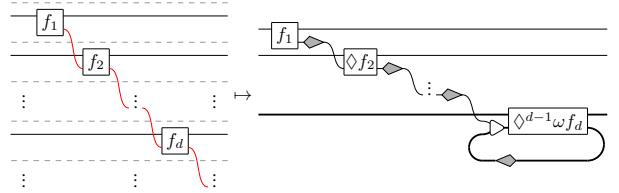
\begin{center}
	\scalebox{0.7}{$\input{gmap0.tikz}\mapsto \input{gmap1.tikz}$}
\end{center}

In fact all diagrams of $\quot{\Stream}{\sim_s}$ can be written in such form.

\begin{lemma}\label{gfull}
	$G:\RegStSeq{\G}\to \quot{\Stream}{\sim_s}$ is full.
\end{lemma}

\subsection{Correspondence with $\quot{\RegStSeq{\G}}{\equiv}$}

We add the following rewriting rules $\equiv$ to $\RegStSeq{\G}$ to match the rewriting rules of $\Stream$. We define $D \equiv D'$ as:
\[ \CM,\IM \vdash D = D' \]
with $\vdash$ admitting the usual deduction rules of a congruence\footnote{The rules for reflexivity, symmetry, and transitivity, together with the preservation under contexts of the form $(S \otimes -)$, $(- \otimes S)$, $(S \circ -)$ and $(- \circ S)$.} together with the coinduction rule:
\begin{center}
	
	$\begin{prooftree}
		\hypo{\forall n\in \mathbb{N}\qquad \Gamma, \left[\D S_{n+1}= \D T_{n+1}\right] \vdash S_n =T_n}
		\infer1[$\left(\D_\mathbb{N}\right)$]{\Gamma \vdash S_0=T_0}
	\end{prooftree}$
	
\end{center}
and with the axioms $\CM$ and $\IM$ being the following, for $c$ causal and $\pi$ idempotent:
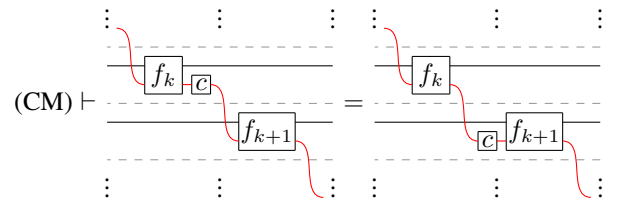
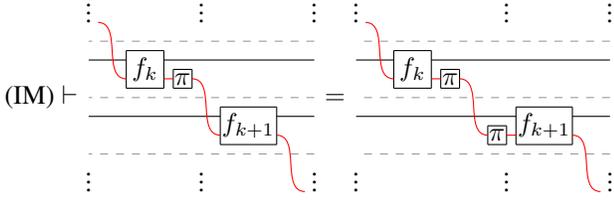
\begin{center}
	$\CM \vdash \input{conggr0.tikz} = \input{conggr1.tikz}$\\[20pt]
	$\IM \vdash \input{congpi0.tikz} = \input{congpi1.tikz}$
\end{center}
We write $\quot{\RegStSeq{\G}}{\equiv}$ for the quotient $\RegStSeq{\G}$ by this congruence.

\begin{lemma}
	$f\equiv g \Leftrightarrow G(f)=G(g)$
\end{lemma}

This implies that $G$ factorizes into the projection $\RegStSeq{\G}\to \quot{\RegStSeq{\G}}{\equiv}$ followed by a full and faithful functor $\mathfrak{G}:\quot{\RegStSeq{\G}}{\equiv}\to\quot{\Stream}{\sim_s}$.

\begin{center}
	\input{factor.tikz}
\end{center}

It follows that we can completely characterize $\Stream$ by the $\equiv$ relation on stateful morphism sequences.

\begin{theorem}\label{thm:equivalence_cat}
	For any discard monochromatic prop $\G$:
	\begin{center}
		$\quot{\RegStSeq{\G}}{\equiv}\quad\simeq\quad\quot{\Stream}{\sim_s}$.
	\end{center}
\end{theorem}

\section{The Semantics of $\Stream$}\label{sec:sem}

In this section, we define the semantics of our language $\Stream$. We assume given a semantics $\sem{-} : \G \to \C$ of $\G$, and will extend it into a semantics $\sembis{-}$ for diagrams of $\Stream$.  The semantics of $D \in \Stream$  will be its finite approximation sequences, \emph{i.e.}, a sequence of morphisms of $\C$, with the $k^\text{th}$ morphism corresponding to the computations done up until the $k^\text{th}$ tick of the clock. We prove that our semantics is sound. We also prove universality and completeness up to some  additional requirements on $\G$ and $\C$.

\subsection{Discard Category}
\label{sec:discard}
The only required assumption for the definition of $\Stream$ is that $\G$ is a \emph{discard prop}. We expect the category $\C$ to enjoy the same property. As a consequence, we introduce a straightforward extension of the notion of discard prop to the symmetric monoidal case:
\begin{definition}[Discard]
	A \textbf{discard} category is a symmetric monoidal category $(\C,\otimes,I)$ together with, for every object $A$, a \textbf{discard} map $\Terre{A} : A \to I$ such that $\Terre{I} = \id_{I}$ and $\Terre{A \otimes B} = \Terre{A} \otimes \Terre{B}$. 
\end{definition}
We write $\C_{\textup{causal}}$ its subcategory of \emph{causal morphisms}, \emph{i.e.}, morphisms $f$ such that $\Terre{B} \circ f = \Terre{A}$. We say that a monoidal functor $\mathcal{F}$ between discard categories is \emph{discard-preserving} if $\mathcal{F}(\Terre{A}) = \Terre{\mathcal{F}(A)}$. We say that it is \emph{discard-reflecting} if whenever $\mathcal{F}(f) = \Terre{\mathcal{F}(A)}$ we have $f = \Terre{A}$. Those properties are equivalent to $\mathcal{F}$ respectively preserving or reflecting causal morphisms.

From now on, we assume that $\C$ is a discard category and the functor $\sem{-}$ is monoidal and discard-preserving, as those properties are required for soundness. For completeness, we will additionally expect $\sem{-}$ to be discard-reflecting.

\subsection{Semantics}

To define the semantics of a diagram $D \in \Stream(a,b)$, we rely on the fact that $\quot{\Stream}{\sim_s}$ and $\quot{\RegStSeq{\G}}{\equiv}$ are equivalent categories (Theorem \ref{thm:equivalence_cat}), and write $\partial D$ the equivalence class of stateful morphism sequences associated to $\quot{D}{\sim_s}$. We then apply $\sem{-}$ to each layer:
we write $\sem{-}^{\textup{St}}$ for the functor from $\quot{\RegStSeq{\G}}{\equiv}$ to $\quot{\RegStSeq{\C}}{\equiv}$ which simply applies $\sem{-}$ to every layer of the sequence. This functor inherits all the properties of $\sem{-}$, and is in particular monoidal and discard-preserving.

We then collect all the operations happening up until the $k^\text{th}$ tick for $k \geq 1$. For $\alpha = (\alpha_n)_{n \geq 1} \in \RegStSeq{\C}(A,B)$ we define its $k$-th finite approximation $\fa{\alpha}_k \in \C(\fa{A}_k,\fa{B}_k)$ as follows:
\[ \fa{A}_k \df A_1 \otimes \dots \otimes A_k \qquad  \fa{\alpha}_k  \df \input{onion_from_strat_bis.tikz}\]
We note that while stateful morphism sequences of a diagram are defined up to the observational equivalence $\equiv$, $\fa{-}_k$ is sound with respect to this congruence:
\begin{lemma}[Soundness]\label{prop:sembis_is_sound}
	Whenever $\alpha \equiv \beta$, for every $k \geq 1$ we have $\fa{\alpha}_k = \fa{\beta}_k$. So $\fa{-}_k$ can be seen as a functor from $\quot{\RegStSeq{\C}}{\equiv}$ to $\C$.
\end{lemma}
We then collect all of the $(\fa{\alpha}_k)_{k \geq 1}$ into a morphism $\fa{\alpha}$ of the category of morphism sequences $\Seq$ defined below, which we will later refine into the category of finite approximation sequences $\Onion$.

\begin{definition}
	For $\C$ a discard category, we define the discard category of sequences of $\C$, written $\Seq$, as follows:
	\begin{itemize}
		\item Its objects are sequences $(A_k)_{k \geq 1}$ with $A_k\in \textup{Obj}(\C)$.
		\item Its morphisms are sequences  $(f_k)_{k \geq 1}$ such that $f_k \in \C(A_k,B_k)$
		\item The composition (resp. monoidal product, resp. discard) is the composition (resp. monoidal product, resp. discard) component-wise.
	\end{itemize} 
\end{definition}

Using Lemma \ref{prop:sembis_is_sound}, we can now define $\sembis{-}$ as the composition of previously defined functors:
\[ \sembis{-} \df \fa{\sem{\parsem[]{-}}^{\textup{St}}} : \Stream \to \Seq \]
It is a discard-preserving monoidal functor.

\subsection{Examples}

\begin{figure}
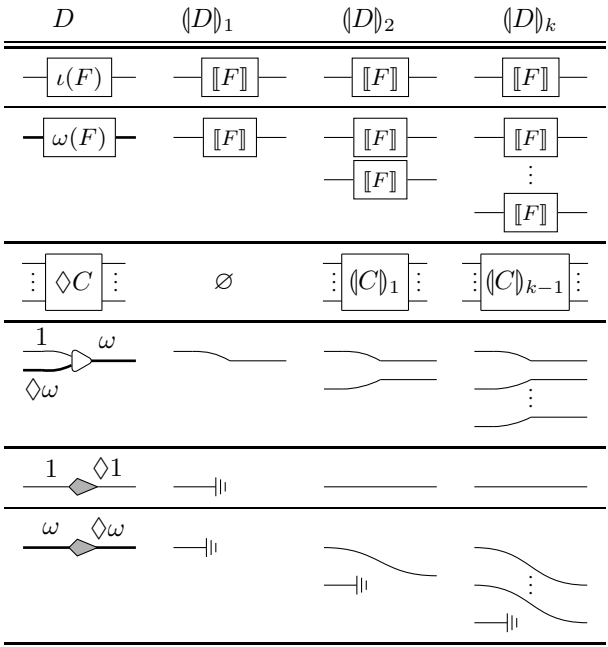

	\renewcommand{\arraystretch}{1.5}
	\begin{tabular}{c}
		$D \qquad\qquad \sembis{D}_1 \qquad\qquad \sembis{D}_2 \qquad\qquad \sembis{D}_k$\\\hline\hline
		$\input{sem_iotaF.tikz}$\\\hline
		$\input{sem_omegaF.tikz}$ \\\hline
		$\input{sem_diamand.tikz}$\\\hline
		$\input{sem_triangle.tikz}$ \\\hline
		$\input{sem_delai.tikz}$\\\hline
		$\input{sem_delai_omega.tikz}$\\\hline
	\end{tabular}
	\caption{Semantics of Generators}\label{fig:sem}
\end{figure}
To illustrate this semantics, we detail in Fig. \ref{fig:sem} the semantics of five special cases: the $\iota$-morphisms, the $\omega$-morphisms, the delayed morphisms, the stream initialization, the $\iota$-delay and the $\omega$-delay. The second case covers one of our first examples in $\CPM$ (see Example \ref{ex:scal}): the measure applied to a completely mixed state, seen as an $\omega$-morphism. At each tick of the clock, it will generate another instance of mixed stated followed by measure, which is the scalar $1/2$. At the $k$-th tick, we then have:
\begin{center} $\sembis{\input{exa1.tikz}}_k = \left(\frac{1}{2}\right)^k$ \end{center}
\subsection{Universality}

While this semantics is sound, $\Seq$ contains significantly more behaviors than $\Stream$, as $\Seq$ does not constrain any relation between the states of the computation at ticks $k$ and $k+1$. To obtain universality, we need to restrict $\Seq$ to a more realistic category. As presented previously in \cref{sec:monotone} in the quantum case, we start by adding a condition of \emph{monotonicity}:
\begin{itemize}
	\item An object $(A_k)_{k \geq 1}$ is monotone if for every $k \geq 1$, there exists an object $A'_k$ of $\C$ such that $A_{k+1} = A_k \otimes A'_k$.
	\item A morphism $(g_k)_{k \geq 1}$ is monotone if for every $k \geq 1$, we have
	\begin{center}
		$\input{monotone_q.tikz} \preceq \input{monotone_q2.tikz}$
	\end{center}
	for $\preceq$ defined below.
\end{itemize}
Decreasing along $\preceq$ means adding additional observable effects. In particular, whenever $\C$ is the category $\UCPM$, or any other category in which every morphism is causal, then $\preceq$ is nothing but  the equality. In the general case, we take $\preceq$ to be the following:
\begin{definition}
	For $f,g \in \C(A,B)$, we have $f \preceq g$ whenever there exist $g_0 \in \C(A,B \otimes X)$ and $f_0 \in \C(X,I)$ such that
	\begin{center}
		$f = \input{g0_f0.tikz} \qquad \qquad g = \input{g0_ground.tikz}$
	\end{center}
\end{definition}
In the case of $\C = \CPM$, $\preceq$ is tightly linked to the Loewner order, as we have
\[ f \preceq g \iff \exists \lambda > 0, (g - \lambda \cdot f) \in \CPM \]
We note that without additional restrictions on $\C$, $\preceq$ might not be transitive. To ensure transitivity, we require for $\C$ to have a notion of \emph{pseudo-purification}, which is a generalization of the notion of purification in the quantum case. From a programming point of view, each time the morphism is ``dumping'' some information through the use of a discard, we want to intercept this discard map and instead output that information on a secondary output so that they can be used for later computations.
\begin{definition}
	In a discard category $(\C,\otimes,I,\Terre{})$, a morphism $p \in \C(A,B \otimes X)$ is said to be a pseudo-purification of $f \in \C(A,B)$, and we write $p \in \Pure(f)$, if for every $g \in \C(A,B \otimes Y)$
	\[ f = \input{purification_G.tikz} \implies \exists c \text{ causal}, g = \input{p_c.tikz} \]	
\end{definition}
In particular, whenever $p \in \Pure(f)$ we always have
\[ f = \input{purification.tikz} \]
and given two pseudo-purifications $p_1,p_2 \in \Pure(F)$ there exist two causal morphisms $c_1$ and $c_2$ such that
\[ p_1 = \input{p2_c2.tikz} \qquad p_2 = \input{p1_c1.tikz} \]
The uniqueness of pseudo-purification up to a causal morphism is an analogue to Stinespring's dilation (Theorem \ref{thm:stinespring}), that we adapted to account for the possible absence of a notion of isometry.
\begin{definition}
	A discard category $(\C,\otimes,I,\Terre{})$ is said to be \emph{pseudo-purifiable} whenever every morphism has a pseudo-purification, and moreover for $f_1 \in \C(A_1,B_1 \otimes C)$, $f_2 \in \C(C \otimes A_2,B_2)$, $p_1 \in \Pure(f_1)$ and $ p_2 \in \Pure(f_2)$ we have:
		\[\input{purification_comp3.tikz} \in \Pure\left(\input{F_G_compo_tenseur.tikz}\right) \]
\end{definition}
In particular for $C = I$ we obtain that:
\[\input{purification_tensor2.tikz} \in \Pure(f_1 \otimes f_2) \]
\begin{lemma}
 	In a pseudo-purifiable category $(\C,\otimes,I,\Terre{})$, $\preceq$ is transitive, is preserved under composition and monoidal product, hence forms a pre-order enrichment of $\C$.
\end{lemma}
\begin{lemma}\label{cppure}
The categories $\CPM$ and $\UCPM$ are pseudo-purifiable categories. Moreover, every purification (for the usual notion of purity, see Section \ref{sec:quantum}) is a pseudo-purification.
\end{lemma}

Not every pseudo-purification in $\CPM$, is a purification. Indeed, a purification of $\id_I$ is any isometric (so pure and causal) states $q \in \CPM(I,A)$, while every causal state (pure or not) is a valid pseudo-purification of $\id_I$.

\begin{lemma}\label{cartpure}
	Every cartesian category $(\C,\times,I)$ is pseudo-purifiable, and for every $f \in \C(A,B)$, the pairing $(f,\id_A) \in \C(A,B \times A)$ is one of the pseudo-purifications of $f$.
\end{lemma}

With monotonicity defined, we are one step closer to universality. However, monotone sequences still contain behaviors that are not captured by our language. More precisely, monotone sequences contain behaviors that would be representable by an infinitely-sized diagram, but cannot be represented in our finitary language. We restrict ourselves to \emph{regular} monotone sequences, which we call \emph{finite approximations} and define as follows:
\begin{definition}\label{def:fa}
	For $\C$ a pseudo-purifiable category, we define the discard category of \textbf{finite approximation sequences} $\Onion$ as $\Seq$ restricted to  
	\begin{itemize}
		\item Objects $(A_k)_{k \geq 1}$ that are monotone (see above) and regular, \emph{i.e.}, there exists $n \geq 1$ and $A'$ such that for all $k \geq n$ we have $A_k = A_n \otimes A'^{\otimes (k-n)}$.
		\item Morphisms $(g_k)_{k \geq 1}$ that are monotone (see above) and regular, \emph{i.e.}, there exists $n \geq 1$, $g^{\textup{irreg}} \in \C(A_n,B_n \otimes M)$ and $g^{\textup{reg}} \in \C(M \otimes A',B' \otimes M)$ such that for all $k \geq n$ 
	\end{itemize}
\begin{center}$g_k = \input{regularity_g.tikz}$.
\end{center}
\end{definition}

\begin{proposition}
	For every $f \in \Stream$, $\sembis{f} \in \Onion$.
\end{proposition}

\begin{theorem}[Universality]\label{thm:universality}
	If $\C$ is pseudo-purifiable and the functor $\sem{-} : \G \to \C$ is full, then $\sembis{-} : \Stream \to \Onion$ is full too.
\end{theorem}

\subsection{Completeness}

The discard allows us to discard ``future computation'' and only keep what happens before a given tick. However, in the proof of completeness, we will need to act dually and have a way to undo the first computations to only keep what happens in later ticks. Ideally, we would want morphisms to be surjections (or epimorphisms), as a surjection $f$ satisfies the following:
\[g \circ f = h\circ f \iff g = h\]
Unfortunately, the category $\CPM$ contains non-surjective morphisms, and even contains morphisms that cannot be decomposed into a surjection followed by an injection\footnote{Though if one allows every Hilbert space instead of restricting ourselves to only Hilbert spaces of dimension a power of two, every morphism becomes decomposable.}. However, for $f \in \CPM(A,B)$ a morphism corresponding to a pure quantum computation, there always is an idempotent morphism $\pi$ (the projector over the image of $f$) such that: 
	\[g \circ f = h \circ f \iff g \circ \pi = h \circ \pi\]
We formalize a slight generalization of this property in the concept of morphism with shadows.
\begin{definition}
	In a symmetric monoidal category $(\C,\otimes,I)$ a morphism with shadows is a morphism $f \in \C(A,B)$ such that for every isomorphism $\iota: B \cong B_0 \otimes B_1$, there exists an idempotent morphism $\pi : B_1 \to B_1$, called shadow of $f$, such that:
	\begin{center}
		$\forall g,h,\quad\begin{matrix} \input{shadow_f_g.tikz} = \input{shadow_f_h.tikz} \\ \iff \\ \input{shadow_pi_g.tikz} = \input{shadow_pi_h.tikz}  \end{matrix}$
	\end{center}
	A pseudo-purifiable category is said to be a shadow pseudo-purifiable category if every morphism has at least one pseudo-purification which is a morphism with shadows.
\end{definition}

The idea of using idempotent morphisms to characterize the image of morphisms is already present in range categories (see \cite{Cocket2012range}), but we chose here a definition much more tailored to our needs than the strict structure of range categories.
\begin{lemma}\label{cpshadow}
	The categories $\CPM$ and $\UCPM$ are shadow pseudo-purifiable. In fact, all the morphisms that are pure quantum computations are morphisms with shadows.
\end{lemma}

\begin{theorem}[Completeness]\label{thm:completeness}
	If $\C$ is a shadow pseudo-purifiable category and the functor $\sem{-} : \G \to \C$ is faithful and discard-reflecting, then $\sembis{-}$ is faithful too.
\end{theorem}

This follows from the faithfulness of 
\[ \Stream \xrightarrow{\partial} \quot{\RegStSeq{\G}}{\equiv} \xrightarrow{\sem{-}^{\textup{St}}} \quot{\RegStSeq{\C}}{\equiv} \xrightarrow{\fa{-}} \Onion \]

\section{Applications}
\subsection{The ZX-calculus}

The ZX-calculus is known to be universal and complete for various fragments of quantum mechanics \cite{pi_2-complete,jeandel2018complete,HNW}, including the most general case: the ZX-calculus -- equipped with a discard map -- is  universal and complete for $\CPM$ \cite{carette2019completeness}. Furthermore Lemmas \ref{cppure} and \ref{cpshadow} give us that $\CPM$ is a pseudo-purifiable shadow category. As a consequence ZX$^\omega$ is universal and complete for monotone regular sequences over $\CPM$. Similar results hold for variants of the ZX-calculus like the ZW- or the ZH-calculi which are also universal and complete for $\CPM$ \cite{carette2019completeness}. On the other hand, quantum circuits, equipped with a discard map, are universal for $\UCPM$ but no axiomatization is known to be complete.

\subsection{The Fragment $\Stream_0$ of Initialized Delays}

Using delayed trace for process with memory is not a new idea. It has appeared in various context \cite{bonchi2014categorical, sprunger2019differentiable,ghica2017diagrammatic,roman2020comb}, usually in the cartesian case and with an initialized delayed trace restricting the expressivity to sequences which are regular from the beginning. In this subsection we present the fragment of our language that corresponds to this situation and then in the following subsections compare our results to selected examples from the literature providing an overview of the generality of our construction and its limitations.

In $\Stream$, the need for delayed types $\D^{n+1} \omega$ arises from the fact that the delay morphism on streams takes as an input a stream but outputs a stream with an undefined behavior for the first tick. However, adding delayed types $\D^{n+1} \omega$ is not the only solution to this problem. Indeed, most preexisting works instead chose to use an \emph{initialized delay}. Those initialized delays can be encoded in $\Stream$ as follows: given $a$ a type of $\G$ and $F \in \G(0,a)$ we define the delay initialized by $F$ as:
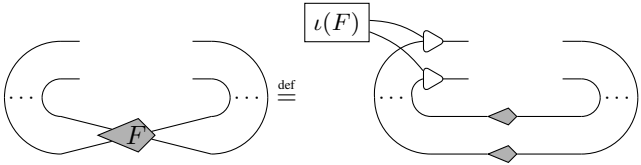
\[ \input{memcell.tikz} \df \input{memcell_def.tikz} \]
We note that in the cartesian case, every morphism $F \in \C(0,a)$ can be decomposed into a product of morphisms on $\C(0,1)$, so one only needs to define the initialized delayed trace on single wires rather than collection of those.

We consider $\Stream_0$ the sublanguage of $\Stream$ where every delay has to be initialized:
\begin{itemize}
	\item Most of the objects of $\Stream$ are unnecessary, we only take 
	\[\textup{Obj}(\Stream_0) \df \{ n \cdot \omega ~|~n \in \NN\} = \omega(\textup{Obj}(\G))\]
	\item $\Stream_0$ is a prop, with for generators the morphisms $\omega(D)$ for $D$ a morphism of $\G$, and the delayed trace initialized by $F \in \G(0,c)$: 
	\[\Dtri{F}{\omega(c)}{a,b}{D} \in \Stream_0(a,b) \text{ for } D \in \Stream_0(a+\omega(c),b+\omega(c)) \]
\end{itemize}

The stateful morphism sequences associated to $\Stream_0$ are sequences $(f_k)_{k \geq 1}$ for which the regularity condition starts at the second tick: $\forall k \geq 2, f_k = f_2$.
Similarly, on the semantics side, the corresponding finite approximations $(f_k)_{k \geq 1}$ are regular starting from the first tick:
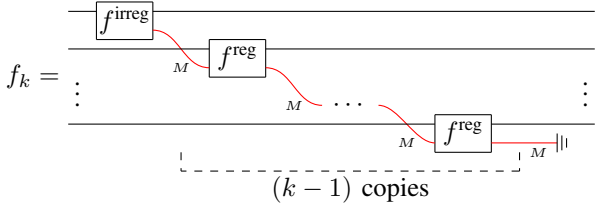
\[ f_k = \input{regularity_imm.tikz} \]

Note that all diagrams in $\Stream_0$ can be rewritten into the form:

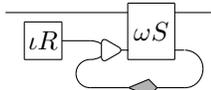
\begin{center}
	$\input{normal.tikz}$
\end{center}

We don't claim any universality or completness result for this fragment.

\subsection{Quantum Channels with Memory}
\label{sec:werner}

The first inspiration of this work was the quantum channels with memory of \cite{kretschmann2005quantum}. This corresponds to the case where we take $\G$ to be the category of quantum circuits with discard. There are still no known complete finite axiomatizations of the full language. Assuming a free axiomatization matching the semantics in $\UCPM$ we can construct $\Stream$. $\UCPM$ is a pseudo-purifiable discard shadow category thus $\Stream_0$ is complete and universal for monotone finite approximations regular from tick $2$. $\UCPM$ is also a {semi-cartesian} category, meaning that all its morphisms are causal. In such situation, we can significantly simplify  the definition of $\Stream$ by fusing the $\left({\ground}\right)$ and $\left(\pi\right)$ rules. Moreover, the order $\preceq$ then collapses to the identity, leading to a clear interpretation of the monotone sequences as processes where the present does not depend on the future.

Applying quantum mechanics to those stream transformers usually requires infinite-dimensionalal Hilbert spaces. In \cite{kretschmann2005quantum}, quasi local algebras are used to represent those processes in the Heisenberg picture, a quantum analog of the predicate transformer point of view. The authors require from their channels a causality condition that matches precisely our monotonicity requirement. The main difference with our work is that they consider streams on $\mathbb{Z}$ (so an infinity of ticks happened before the tick $1$) while we consider streams on $\mathbb{N}_{\geq 1}$.

However, given a fixed quantum state for the ticks $(-\infty,0]$, their condition of translational invariance corresponds to our condition of regularity on finite approximations. They obtain a structure theorem that corresponds precisely to the general form of diagrams in $\Stream_0$.

\subsection{The Cartesian Case}

The stateful morphism sequences were first defined by \cite{sprunger2019differentiable} in the cartesian case. A cartesian category is always pseudo-purifiable and semi-cartesian.  

In \cite{sprunger2019differentiable}, the authors build a category very similar to $\Stream_0$ by quotienting stateful morphism sequences by an observational equivalence relation corresponding to ours in the cartesian case. They define the exact same initialized delayed trace and study in more details the category of stateful sequences of morphisms. However they do not show any universality or completeness results, focusing instead on differentiability.

Note that taking $\G$ to be boolean circuits with semantics in \textbf{Set} (which is a cartesian shadow category) we can deduce from $\cite{sprunger2019differentiable}$ that $\Stream_0$ has the expressive power of Mealy machines. In this direction, further work will focus on understanding exactly which kind of synchronous circuits can be represented by our construction in connection to the work of \cite{ghica2017diagrammatic}.  

\subsection{Signal Flow Graphs}\label{sec:lin}

Another work  similar to ours is the work of \cite{bonchi2014categorical} on signal flow graphs. Similarly to \cite{kretschmann2005quantum}, they consider streams on $\mathbb{Z}$ while we consider streams on $\mathbb{N}_{\geq 1}$. There is however another major difference in approach: they represent streams and their operations as a whole (using power series) rather than through their finite approximations. This leads to a set of axioms incompatible to ours, in particular the rule (S5) of their Definition 3 would translate to the following:
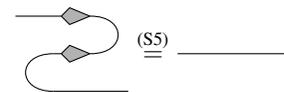
\[ \input{unsound_rule.tikz} \stackrel{\text{(S5)}}{=} \input{identity.tikz}\]
which is unsound for finite approximations. Indeed, the left hand side is interpreted by us as follows: the information $i$ received at the tick $n$ is not immediately output, instead the system generates a ``blank'' output (the transposed of $\Terre{}$), which retroactively changed to be equal to $i$ at tick $n+1$. This behavior is fundamentally non-causal, but is expected as the co-unit of a compact closure is not a causal morphism.

However, when considering the fragment $\mathbb{SF}$ of circuits that only contain initialized guarded traces (which they call feedbacks), we recover a correspondence. In fact their calculus seems to be exactly $\Stream_0$ when we take $\G$ to be the graphical language $\mathbb{HA}$.

Taking the same example as in their paper, we can describe the Fibonacci sequence as a morphism of $\Stream_0$, with $\G$ being the prop of linear operations on tuples of integers:
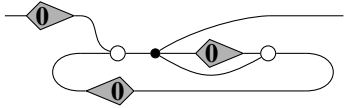
\begin{center}
	$\input{fibbo.tikz}$
\end{center}
where $\bfup{0} \in \G(0,1)$ being the integer zero, white dots representing the addition and black dots representing the copy.
On the input stream $1,0,0$, \emph{etc} this circuit will output the Fibonacci sequence $0,1,1,2,3$, \emph{etc}.

Another interesting connection with this line of work is to take $\G$ to be $\mathbb{IH}$, a graphical calculus which have been shown to be complete for linear relations \cite{zanasi2018interacting}. The order relation $\preceq$ then coincides with the subspace relation for vector spaces. 

More work has still to be done along this line to unravel all the connections between the two formalisms.

\section*{Acknowledgment}
This work is funded by ANR-17-CE25-0009 SoftQPro, ANR-17-CE24-0035 VanQuTe,
PIA-GDN/Quantex, and LUE/UOQ. This research is also supported by the project NEASQC funded from the European Union’s Horizon 2020
research and innovation programme (grant agreement No 951821).

\bibliographystyle{IEEEtran}
\bibliography{synchlics}

\newpage
\appendix
\input{synchlicsApp}

\end{document}

%% file: sample.tikzstyles
\tikzstyle{gn}=[fill=green, draw=black, shape=circle,  tikzit category=ZX, tikzit fill=green, tikzit draw=black, tikzit shape=circle, inner sep=0.1em]
\tikzstyle{rn}=[fill=red, draw=black, shape=circle, tikzit fill=red, tikzit draw=black, tikzit category=ZX, tikzit shape=circle, inner sep=0.1em]
\tikzstyle{divide}=[regular polygon, regular polygon sides=3, shape border rotate=90, draw=black,fill=white, inner sep=2pt, tikzit category=scal, rounded corners=0.8mm]
\tikzstyle{black}=[fill=black, draw=black, shape=circle, tikzit fill=black, tikzit draw=black, tikzit shape=circle, tikzit category=IH, inner sep=1.2pt]
\tikzstyle{gather}=[fill=white, draw=black, tikzit category=scal, rounded corners=0.8mm, regular polygon, regular polygon sides=3, shape border rotate=-90, inner sep=2pt]
\tikzstyle{ggen}=[fill=white, draw=black, shape=rectangle, rounded corners=2mm,  line width=1pt, tikzit draw=red, tikzit category=scal]
\tikzstyle{white}=[fill=white, draw=black, shape=circle, inner sep=2pt, tikzit category=IH]
\tikzstyle{mbox}=[fill=white, draw=black, rounded rectangle, rounded rectangle west arc=none, tikzit category=scal, tikzit shape=rectangle]
\tikzstyle{A}=[fill=white, shape=circle, tikzit category=scal, inner sep=1pt]
\tikzstyle{ggreen}=[fill=green, draw=black, shape=circle, tikzit category=SZX, tikzit fill=green, tikzit draw=black, line width=1pt, inner sep=0.1em]
\tikzstyle{gred}=[fill=red, draw=black, shape=circle, rounded corners=2mm,  tikzit category=SZX, inner sep=0.1em, tikzit fill=red, line width=1pt]
\tikzstyle{ghad}=[fill=yellow, draw=black, shape=rectangle, tikzit category=SZX, tikzit shape=rectangle, tikzit fill=yellow, inner sep=0.1em, line width=1pt]
\tikzstyle{boxm}=[fill=white, draw=black, rounded rectangle, tikzit category=scal, tikzit shape=rectangle, rounded rectangle east arc=none]
\tikzstyle{box}=[fill=white, draw=blue, shape=rectangle]
\tikzstyle{had}=[fill=yellow, draw=black, shape=rectangle, tikzit category=ZX, tikzit fill=yellow, tikzit draw=black, inner sep=0.1em]
\tikzstyle{gwhite}=[fill=white, draw=black, shape=circle, tikzit fill=white, tikzit shape=circle, line width=1 pt, inner sep=2 pt, tikzit draw=red]
\tikzstyle{gblack}=[fill=black, draw=black, shape=circle, tikzit fill=black, tikzit shape=circle, line width=1 pt, inner sep=2 pt, tikzit draw=red]
\tikzstyle{antipode}=[fill=red, draw=black, shape=rectangle, tikzit fill=red, tikzit draw=black, tikzit shape=rectangle, inner sep=2pt]
\tikzstyle{rdelay}=[kite, draw, kite vertex angles=90 and 45,rotate=90,scale=0.6,fill=gray!60!white]
\tikzstyle{ldelay}=[kite, draw, kite vertex angles=90 and 45,rotate=-90,scale=0.6,fill=gray!60!white]
\tikzstyle{mongr}=[fill=green, draw=green, shape=circle, inner sep=2pt]
\tikzstyle{monbl}=[fill=blue, draw=blue, shape=circle, inner sep=2pt]
\tikzstyle{bg}=[inner sep=0.7mm, minimum width=0pt, minimum height=0pt, fill=green, draw=white, very thick, shape=circle]
\tikzstyle{br}=[inner sep=0.7mm, minimum width=0pt, minimum height=0pt, fill=red, draw=white, very thick, shape=circle]
\tikzstyle{rmat}=[draw, signal, fill=zx_grey, signal to=east, signal from=west, inner sep=1pt, minimum height=6pt]
\tikzstyle{lmat}=[draw, signal, fill=zx_grey, signal to=west, signal from=east, inner sep=1pt, minimum height=6pt]
\tikzstyle{umat}=[draw, signal, fill=zx_grey, signal to=north, signal from=south, inner sep=1pt, minimum width=6pt]
\tikzstyle{dmat}=[draw, signal, fill=zx_grey, signal to=south, signal from=north, inner sep=1pt, minimum width=6pt]
\tikzstyle{arrow}=[->]
\tikzstyle{very thick}=[-, line width=1pt, tikzit draw=red]
\tikzstyle{pointille}=[dashed, -]
\tikzstyle{memo}=[dashed, -, draw=gray!80!white]
\tikzstyle{red}=[-, draw=red]
\tikzstyle{blue}=[-, draw=blue]
\tikzstyle{green}=[-, draw=green]
\tikzstyle{yellow}=[-, draw=orange]
\tikzstyle{arrow}=[->]
\tikzstyle{strike}=[-, tikzit draw={rgb,255: red,191; green,0; blue,64}, strike through]
\tikzstyle{strike'}=[-, tikzit draw=cyan, strike bend]

\tikzstyle{box2}=[shape=rectangle, text height=6ex, text depth=0.25ex, yshift=0mm, fill=blue, draw=black, minimum height=1mm, minimum width=5mm, font={\small}]
\tikzstyle{box}=[shape=rectangle, text height=1.5ex, text depth=0.25ex, yshift=0.5mm, fill=white, draw=black, minimum height=5mm, yshift=-0.5mm, minimum width=5mm, font={\small}]
\tikzstyle{Z dot}=[inner sep=0mm, minimum size=2mm, shape=circle, draw=black, fill={rgb,255: red,160; green,255; blue,160}]
\tikzstyle{gdot}=[minimum size=3mm, font={\scriptsize\boldmath}, shape=rectangle, rounded corners=1.3mm, inner sep=1mm, outer sep=-1.8mm, scale=0.8, tikzit shape=circle, draw=black, fill=green, tikzit draw=blue]
\tikzstyle{X dot}=[Z dot, shape=circle, draw=black, fill={rgb,255: red,220; green,0; blue,0}]
\tikzstyle{rdot}=[minimum size=3mm, font={\scriptsize\boldmath}, shape=rectangle, rounded corners=1.3mm, inner sep=1mm, outer sep=-1.8mm, scale=0.8, tikzit shape=circle, draw=black, fill=red, tikzit draw=blue]
\tikzstyle{grdot}=[minimum size=3mm, font={\scriptsize\boldmath}, shape=rectangle, rounded corners=1.3mm, inner sep=1mm,
 line width=1pt, outer sep=-1.5mm, scale=0.8, tikzit shape=circle, draw=black, fill=red, tikzit draw=blue]
 \tikzstyle{ggdot}=[minimum size=3mm, font={\scriptsize\boldmath}, shape=rectangle,  line width=1pt, rounded corners=1.3mm, inner sep=1mm, outer sep=-1.5mm, scale=0.8, tikzit shape=circle, draw=black, fill=green, tikzit draw=blue]

%% file: figures/ket0.tikz
\begin{tikzpicture}
	\begin{pgfonlayer}{nodelayer}
		\node [style=none] (41) at (1, 0) {};
		\node [style=white] (42) at (0, 0) {};
	\end{pgfonlayer}
	\begin{pgfonlayer}{edgelayer}
		\draw (42) to (41.center);
	\end{pgfonlayer}
\end{tikzpicture}

%% file: figures/bra0.tikz
\begin{tikzpicture}
	\begin{pgfonlayer}{nodelayer}
		\node [style=none] (41) at (0, 0) {};
		\node [style=white] (42) at (1, 0) {};
	\end{pgfonlayer}
	\begin{pgfonlayer}{edgelayer}
		\draw (42) to (41.center);
	\end{pgfonlayer}
\end{tikzpicture}

%% file: figures/cup.tikz
\begin{tikzpicture}
	\begin{pgfonlayer}{nodelayer}
		\node [style=none] (0) at (0.5, 0.5) {};
		\node [style=white] (1) at (0, 0) {};
		\node [style=none] (2) at (0.5, -0.5) {};
	\end{pgfonlayer}
	\begin{pgfonlayer}{edgelayer}
		\draw [style=red, bend right=45] (0.center) to (1);
		\draw [style=red, bend right=45] (1) to (2.center);
	\end{pgfonlayer}
\end{tikzpicture}

%% file: figures/cap.tikz
\begin{tikzpicture}
	\begin{pgfonlayer}{nodelayer}
		\node [style=none] (0) at (0, 0.5) {};
		\node [style=white] (1) at (0.5, 0) {};
		\node [style=none] (2) at (0, -0.5) {};
	\end{pgfonlayer}
	\begin{pgfonlayer}{edgelayer}
		\draw [style=red, bend left=45] (0.center) to (1);
		\draw [style=red, bend left=45] (1) to (2.center);
	\end{pgfonlayer}
\end{tikzpicture}

%% file: figures/snake1.tikz
\begin{tikzpicture}
	\begin{pgfonlayer}{nodelayer}
		\node [style=none] (8) at (1.25, 0) {};
		\node [style=none] (9) at (0, 0) {};
	\end{pgfonlayer}
	\begin{pgfonlayer}{edgelayer}
		\draw [style=red, in=0, out=180] (8.center) to (9.center);
	\end{pgfonlayer}
\end{tikzpicture}

%% file: figures/exp1.tikz
\begin{tikzpicture}
	\begin{pgfonlayer}{nodelayer}
		\node [style=none] (27) at (0, 0) {};
		\node [style=none] (28) at (2, 0) {};
	\end{pgfonlayer}
	\begin{pgfonlayer}{edgelayer}
		\draw [style=very thick, in=180, out=0] (27.center) to (28.center);
	\end{pgfonlayer}
\end{tikzpicture}

%% file: figures/delay1.tikz
\begin{tikzpicture}
	\begin{pgfonlayer}{nodelayer}
		\node [style=none] (40) at (2, 0) {};
		\node [style=none] (41) at (0, 0) {};
		\node [style=rdelay] (43) at (1, 0) {};
	\end{pgfonlayer}
	\begin{pgfonlayer}{edgelayer}
		\draw (41.center) to (43);
		\draw (43) to (40.center);
	\end{pgfonlayer}
\end{tikzpicture}

%% file: figures/identity.tikz
\begin{tikzpicture}
	\begin{pgfonlayer}{nodelayer}
		\node [style=none] (0) at (0, 0) {};
		\node [style=none] (1) at (3, 0) {};
	\end{pgfonlayer}
	\begin{pgfonlayer}{edgelayer}
		\draw (0.center) to (1.center);
	\end{pgfonlayer}
\end{tikzpicture}

%% file: synchlicsApp.tex

\label{proofs}

\subsection{Equivalence of $\Stream$ and $\RegStSeq{\G}$}
\setcounter{storedvalue}{\thelemma}
\setcounter{lemma}{0}

\begin{lemma}
	Given a type $a$ of degree $d$. There is always a disjoint map $a\to \overline{a}$, where $\overline{a}$ is stratified of degree $d$, as well as a disjoint map $\overline{a}\to \delta^k\overline{a}$ for each $k\geq d$. And if there is a disjoint map $a\to b$ then it is unique and there is a $k$ such that $\delta^k \overline{a}=\delta^k\overline{b}$.
\end{lemma}

\begin{proof}
	Given any type $a$ of degree $d$ we can always use derivatives to construct a map $a\to a'$ where all $\omega$ colors in $a'$ are delayed $d$ times. Then we can use disjoint swaps to reorder each colors and construct a disjoint map $a'\to \overline{a}$ where $\overline{a}$ is stratified of degree $d$. Furthermore using derivatives we can construct a disjoint map $\overline{a}\to \delta^k\overline{a}$ for any $k\geq d$.

	Let $a$ and $b$ be respectively of degree $l_a$ and $l_b$ and let $k\df \max(l_a,l_b)$. We assume there is a disjoint map $f:a\to b$. We now there are disjoint maps of types $a\to \delta^k\overline{a}$ and $b\to \delta^k\overline{b}$. Those maps being invertible, there is a bijection between disjoint maps $a\to b$ and disjoint maps $\delta^k\overline{a}\to \delta^k\overline{b}$. So we just have to show that the map $f':\delta^k\overline{a}\to \delta^k\overline{b}$ is unique. Using the equations $(\triangleleft\triangleright)$, $(\triangleright\triangleleft)$ and the naturality of the swaps we can rewrite $f'=i\circ \sigma \circ d$ where $i$ is made only of initializations, $d$ only made of derivatives and $\sigma$ is a disjoint permutation. Since $i$ and $d$ cannot change the order of the colors and since the colors are already well ordered in $\delta^k\overline{a}$ and $\delta^k\overline{b}$ the disjoint permutation $\sigma$ can only be the identity. $\delta^k\overline{a}$ and $\delta^k\overline{b}$ having the same degree $k$ we deduce that $i=d^{-1}$. So $\delta^k\overline{a}=\delta^k\overline{b}$ and $f'=id_{\delta^k\overline{a}}$. It follows that $f$ is the unique disjoint map $a\to b$.
\end{proof}

\begin{lemma}
	$G:\textup{RegSt}\left(\G\right)\to \quot{\Stream}{\sim_s}$ is full.
\end{lemma}

\begin{proof}
	Let $D:a\to b$ be a diagram of $\Stream$ where $a$ and $b$ are stratified types. Let $d$ be the biggest integer such that there is a wire of type $\D^d 1$ or $\D^d \omega$ appearing in $D$. The degree of $a$ and $b$ must be less than $d$ so there is a diagram $D':\delta^d a\to \delta^d b$ such that $D\sim_s D'$. 
	
	Using the axioms of delayed trace we take them out of the diagrams to obtain something of the form:
	
	\begin{center}
		$\input{nf0.tikz}$
	\end{center}
	
	We use $\left(\triangleleft \triangleright\right)$ and $\left(\triangleright\right)$ until all the delayed trace on types $\D^{k+1}\omega$ and $\D^{k}\omega$ becomes on types 
	$\D^{d+1}\omega$ and $\D^{d}\omega$.
	
	\begin{center}
		$\input{nf0.tikz}\stackrel{\left(\triangleleft \triangleright\right)}{=}\input{nf1.tikz}\stackrel{\left(\triangleright\right)}{=}\input{nf2.tikz}$
	\end{center}
	
	Then since $d$ was the biggest delay appearing in the diagram before we know that all the delayed trace on types $\D^{k}\omega$ are in the situation:
	
	\begin{center}
		$\input{nf3.tikz}$
	\end{center}
	
	Now using $\left(\triangleleft \triangleright\right)$ on all wires of type $\D^k\omega$ with $k<d$, and then using $\left(\blacktriangleright\right)$ and $\left(\blacktriangleleft\right)$ we can ensure that the only generators of the form $\D^k\omega G$ are in fact of the form $\D^d \omega G$.
	
	\begin{center}
		\scalebox{0.8}{$\input{nf4.tikz}\stackrel{\left(\triangleleft \triangleright\right)}{=}\input{nf5.tikz}\stackrel{\left(\blacktriangleright\right)}{=}\input{nf6.tikz}$}
	\end{center}
	
	Then applying $\left(\triangleright \triangleleft\right)$ we remove all the remaining wires of type $\D^k\omega $ with $k<d$. Now the only $\D^k\omega $ wires in the diagrams are the $\D^{d}\omega$ and the $\D^{d+1}\omega$ in the delayed trace, so there are no derivations left in the diagram and the only initialisations left are the one connected to delayed traces. We can now group together the generators $\D^k \iota g$ and $\D^d \omega g$ of same type. This gives a diagram of the form: 
	
	\begin{center}
		\scalebox{0.7}{$\input{nf7.tikz}=\input{nf8.tikz}$}
	\end{center}
	
	Which is the image of a regular stateful morphism sequence. So for each diagram $D$ there is a diagram $D'$ such that $D\sim_s D'$ and $D'$ is the image of a stateful morphism sequence. In other words, $G:\textup{RegSt}\left(\G\right)\to \quot{\Stream}{\sim_s}$ is full.
\end{proof}

\begin{lemma}
	$f\equiv g \Leftrightarrow G(f)=G(g)$
\end{lemma}

\begin{proof}
	We want to prove
	\[ \AX \vdash G(f) = G(g) \iff \CM,\IM \vdash f = g\]
	For that, we simply match the set of derivation trees on the left hand side to the set of derivation trees on the right hand side.
	We first look at the deduction rules, which are the rules of a congruence (reflexivity, symmetry, transitivity, composition, monoidal product) plus the coinduction rule. There is a one-to-one correspondence between the rules on both side, relying on $G(\D f) = \D G(f)$ for the correspondence between the two coinduction rules. We just need to match the axioms:
	\begin{itemize}
		\item All the axioms of $\AX$ but $(\ground),(\pi),(\sigma)$ correspond to tautologies.
		\item $(\ground)$ is exactly matched to $\CM$.
		\item $(\pi)$ is exactly matched to $\IM$.
		\item $(\sigma)$ correspond to either a tautology or $\CM$ depending on whether the permuted types are disjoints or not.
	\end{itemize}
\end{proof}
\subsection{Pseudo-Purifiable and Shadow Categories}

\setcounter{lemma}{4}

\begin{lemma}
	In a pseudo-purifiable category $(\C,\otimes,I,\Terre{})$, $\preceq$ is transitive, is preserved under composition and monoidal product, hence forms a pre-order enrichment of $\C$.
\end{lemma}
\begin{proof}
	The preservation under composition and monoidal product are direct to prove. For the transitivity, we assume that $f \preceq g \preceq h$. Using the definition of $\preceq$, we obtain $f_0,g_0,g_1,h_1$ such that:
	\begin{center}
		$f = \input{g0_f0.tikz} \qquad \qquad g = \input{g0_ground.tikz}$\\
		$g = \input{h1_g1.tikz} \qquad \qquad h = \input{h1_ground.tikz}$
	\end{center}
	We then take $f_0',g_0',g_1',h_1'$ pseudo-purifications of $f_0,f_0,g_1,h_1$. In particular, we have $g_0' \in \Pure(g)$. Since $\C$ is pseudo-purifiable, we know that we also have
	\[ \input{h1_g1_pure.tikz} \in \Pure(g) \]
	Using uniqueness of pseudo-purifications up to a causal morphism, we know that there exists $c$ causal such that:
	\[ g_0' = \input{h1_g1_pure_c.tikz} \]
	It follows that
	\[ f = \input{h1_g1_pure_c_f0.tikz} \]
	Hence $f \preceq h$.
\end{proof}

\begin{lemma}
	The categories $\CPM$ and $\UCPM$ are pseudo-purifiable categories. Moreover, every purification (for the usual notion of purity, see Section \ref{sec:quantum}) is a pseudo-purification.
\end{lemma}
\begin{proof}
	We start by showing the second property. We consider a purification of $f$, which is a decomposition $f = (\id \otimes \Terre{}) \circ p$ with $p$ a pure quantum computation, and want to prove that $p \in \Pure(f)$. In other words, for every $g$ such that $f = (\id \otimes \Terre{}) \circ g$, we have to find a causal morphism $c$ such that $g = (\id \otimes c) \circ p$.
	
	To do so, we purify $g$ into $g = (\id \otimes \Terre{}) \circ q$, with $q$ a pure quantum computation. We note that both $p$ and $q$ are purifications of $f$, so using \cref{thm:stinespring}, we know that there exists an isometry $v$ such that $q = (\id \otimes v) \circ p$ or and isometry $v'$ such that $p = (\id \otimes v') \circ q$.
	
	In the former case, we have $g = (\id \otimes \Terre{}) \circ q = (\id \otimes \Terre{}) \circ  (\id \otimes v) \circ p$, so we can take $c = \Terre{} \circ v$, which is causal.
	
	In the latter case, we use the fact that every isometry can be seen as an injection followed by a change of basis, and decompose the isometry $v' \in \CPM(A,B)$ with $\dim(A) = 2^n$ and $\dim(B) = 2^{n+k}$ into $v' = u \circ (\id \otimes \ket{0^k})$ for some unitary $u$. So we have
	$(\id_A \otimes \Terre{\CC^{2^k}}) \circ u^\dagger \circ v' = \id_A$, and it follows that:
	\[ g = (\id \otimes \Terre{}) \circ (\id_A \otimes \Terre{\CC^{2^k}}) \circ u^\dagger \circ p\] 
	We can take $c = (\id \otimes \Terre{}) \circ (\id_A \otimes \Terre{\CC^{2^k}}) \circ u^\dagger$, which is causal.
	
	We now want to prove that $\CPM$ and $\UCPM$ are pseudo-purifiable. Since every morphism has a purification, and every purification has a pseudo-purification, then every morphism has a pseudo-purification. In fact, the set of pseudo-purifications of $f$ is exactly the set of morphisms $g$ that are linked to a purification $p$ in both following ways:
	\begin{itemize}
		\item $g$ can be decomposed into $g = (\id \otimes c) \circ p$ with $c$ a causal morphism
		\item $p$ can be decomposed into $p = (\id \otimes c') \circ g$ with $c'$ a causal morphism
	\end{itemize}
	Using this characterisation, and the fact that purity is preserved under composition and tensor, the preservation of pseudo-purification under composition is direct, hence $\CPM$ and $\UCPM$ are pseudo-purifiable categories.
\end{proof}

\begin{lemma}
	Every cartesian category $(\C,\times,I)$ is pseudo-purifiable, and for every $f \in \C(A,B)$, the pairing $(f,\id_A) \in \C(A,B \times A)$ is one of the pseudo-purifications of $f$.
\end{lemma}
\begin{proof}
	We will show that for $f \in \C(A,B)$, $\Pure(f)$ is exactly the set of $(f,g)$ with $g \in \C(A,C)$ left-invertible, \emph{i.e.}, there exists $h \in \C(C,A)$ such that $h \circ g = \id_A$. We take $f,g,h$ as such.
	
	We consider a morphism $p \in \C(A,B \times X)$ such that $(A \times \Terre{X}) \circ p = f$. Since $\C$ is a cartesian category, we know that there exists $f',g'$ such that $p = (f',g')$, and we immediately have $f = f'$. We have
	\[ p = (f,g') = (B \times (g' \circ h)) \circ (f,g) \]
	Since very morphism is causal, $g' \circ h$ is causal, and it follows that $(f,g)$ is indeed a pseudo-purification of $f$. The stability of pseudo-purifications under composition is direct.	
\end{proof}

\begin{lemma}
	The categories $\CPM$ and $\UCPM$ are shadow pseudo-purifiable. In fact, all the morphisms that are pure quantum computations are shadowful morphisms.
\end{lemma}
\begin{proof}
	While a simpler proof exists for $\CPM$, we present here a proof that holds for both $\CPM$ and $\UCPM$. This means that we cannot take the shadows $\pi$ to be orthonormal projectors, as orthonormal projectors are not trace preserving (except for the identity).
	
	We start by treating the special case of the zero morphism $0_{A,B} \in \CPM(A,B)$. The morphism $0_{B_1,B_1} \in \CPM(B_1,B_1)$ is always a valid shadow, hence $0_{A,B}$ is shadowful.
	
	As isomorphisms are pure quantum computation, to prove that all the pure quantum computations are shadowful, it is enough to prove the following: for $f \in \CPM(A,B_0 \otimes B_1)$ which is a pure quantum computation and not the zero morphism, there exists $\pi \in \UCPM(B_1,B_1)$ such that:	
	\begin{center}
		$\forall g,h,\quad\begin{matrix} \input{shadow_simple_f_g.tikz} = \input{shadow_simple_f_h.tikz} \\ \iff \\ \input{shadow_simple_pi_g.tikz} = \input{shadow_simple_pi_h.tikz}  \end{matrix}$
	\end{center}
	To do so, we proceed as follows.
	Since $f$ is a pure quantum computation, we can consider its corresponding morphism $F \in \bfup{Hilb}(A,B_0 \otimes B_1)$. We consider its image $\textup{im}_{B_1}(F)$ projected on $B_1$, which is a subspace of $B_1$, and a non-trivial one as $f$ is not the zero morphism. We choose an orthonormal basis $(\beta_1,\dots,\beta_k)$ of $\textup{im}_{B_1}(F)$ (with $k \geq 1$), and complete it into an orthonormal basis $(\beta_1,\dots,\beta_n)$ of $B_1$.
	We define the morphism $P \in \bfup{Hilb}(B_1,B_1 \otimes B_1)$ as the unique linear operator such that
	\[ \forall 1 \leq i \leq k, P(\beta_i) = \beta_i \otimes \beta_1 \]
	\[ \forall k+1 \leq i \leq n, P(\beta_i) = \beta_1 \otimes \beta_i \]
	Since $P$ is isometric, it corresponds to a morphism $p \in \UCPM(B_1,B_1 \otimes B_1)$, we now define $\pi := (B_1 \otimes \Terre{B_1}) \circ p$. Intuitively, $\pi$ projects everything which is not in the image of $f$ to $\beta_1$, while being the identity on the image of $f$. Formally, each of the $\beta_i$ corresponds to a pure morphism $b_i \in \UCPM(I,B_1)$ and we have
	\[\forall 1 \leq i \leq k, \pi \circ b_i = b_i \]
	\[\forall  k+1 \leq i \leq n, \pi \circ b_i = b_1 \]
	
	The morphism $\pi$ is idempotent, and since $\pi \circ f = f$, the upward direction of our equivalence is immediate. Proving the downward direction is more complex. 
	
	We consider $g,h$ morphisms of $\CPM$ such that
	\[ \input{shadow_simple_f_g.tikz} = \input{shadow_simple_f_h.tikz}\]
	Using \cref{cppure}, we know that we can find pseudo-purifications which are pure quantum computations, so let $g',h'$ be such pseudo-purifications. Since $f$ is already pure, we can use \cref{thm:stinespring} on $g' \circ f$ and $h' \circ f$, and find an isometry $v$ such that $(\id \otimes v) \circ g' \circ f = h' \circ f$, or an isometry $v'$ such that $g' \circ f = (\id \otimes v) \circ h' \circ f$. Without loss of generality, we assume we have the former.
	
	We write $F,G',H',V$ for the morphisms of $\bfup{Hilb}$ corresponding to $f,g',h'$ and $v$ respectively. We have:\\
	
	\scalebox{0.9}{$\begin{matrix}
		&  \input{shadow_simple_f_g.tikz}& =& \input{shadow_simple_f_h.tikz} \\&&\Downarrow&\\
	& \input{shadow_simple_f_g_prime.tikz} &=& \input{shadow_simple_f_h_prime.tikz}  \\&&\Downarrow&\\
	& \input{shadow_simple_f_g_prime_Hilb.tikz} &=& \input{shadow_simple_f_h_prime_Hilb.tikz}  \\&&\Downarrow&\\
	\forall 1 \leq i \leq k,&  \input{shadow_simple_beta_g.tikz} &=& \input{shadow_simple_beta_h.tikz} \\&&\Downarrow&\\
		\forall 1 \leq i \leq k,& \input{shadow_simple_beta_pi_g.tikz} &=& \input{shadow_simple_beta_pi_h.tikz} \\&&\Downarrow&\\
		\forall 1 \leq i \leq n,& \input{shadow_simple_beta_pi_g.tikz} &=& \input{shadow_simple_beta_pi_h.tikz} \\&&\Downarrow&\\
		& \input{shadow_simple_p_g_Hilb.tikz} &=& \input{shadow_simple_p_h_Hilb.tikz} \\&&\Downarrow&\\
		 & \input{shadow_simple_p_g.tikz} &=& \input{shadow_simple_p_h.tikz} \\&&\Downarrow&\\
		 & \input{shadow_simple_pi_g.tikz} &=& \input{shadow_simple_pi_h.tikz} \\
	\end{matrix}$}
\end{proof}

\subsection{Soundness}

We note that $\CM,\IM \vdash \alpha = \beta$ means by definition $\alpha \equiv \beta$.
\setcounter{lemma}{3}
\begin{lemma}[Soundness]\label{app:lem:soundness_proof}
	Whenever $\CM,\IM \vdash \alpha = \beta$, for every $k \geq 1$ we have $\fa{\alpha}_k = \fa{\beta}_k$
\end{lemma}
\setcounter{lemma}{\thestoredvalue}
\begin{proof}
	For $n \geq 0$ with define the congruence ``equal up until the $n$-th tick'' $\approx_n$ on $\RegStSeq{\C}$ as follows:
	\[ \alpha \approx_n \beta \iff \forall k \leq n, \fa{\alpha}_k = \fa{\beta}_k \]
	In particular we always have $\alpha \approx_0 \beta$. The property we are trying to prove is equivalent to: 
	\[\forall n \geq 0, \left[\CM,\IM \vdash \alpha = \beta\right] \implies \alpha \approx_n \beta \]
	
	The start by showing that the axioms $\CM$ and $\IM$ are indeed sound:
	\begin{itemize}
		\item \CM~Moving a causal from the $k$-th layer to the $(k+1)$-th layer trivially preserve all the finite approximations but the $k$-th one. Looking at the $k$-th finite approximation, and we observe the following:
		\begin{center}
			$\begin{matrix} \fa{\alpha}_k &=& \input{onion_from_strat_bis_plus_causal.tikz}\\
				&= & \input{onion_from_strat_bis.tikz} & = & \fa{\beta}_k \end{matrix}$
		\end{center}
		\item \IM~Duplicating an idempotent from the $k$-th layer to the $(k+1)$-th layer does not change any of the associated finite approximations.
	\end{itemize}
	We now show that whenever $\Gamma \vdash \alpha = \beta$, if for all $\left[ \alpha' = \beta'\right] \in \Gamma$ we have $\alpha' \approx_{n} \beta'$ then we have $\alpha \approx_{n} \beta$.
	
	We take a derivation sequence of $\Gamma \vdash \alpha = \beta$, and proceed by induction on this derivation sequence. 
	\begin{itemize}
		\item The initialisation is the axiom rule $\Gamma \vdash \alpha = \beta$ with $\left[ \alpha = \beta\right] \in \Gamma$. The result is immediate.
		\item The reflexivity, symmetry, transitivity rules are trivial, and the composition and tensor rules correspond to the fact that $\fa{-}_k$ is a monoidal functor.
		\item We now consider the coinduction rule
		\begin{center}
			\begin{tabular}{c}
				$(\forall k \geq 0) \qquad  \Gamma,\left[\D \alpha^{(k+1)} = \D \beta^{(k+1)}\right] \vdash \alpha^{(k)} = \beta^{(k)}$ \\\hline
				$\Gamma \vdash \alpha^{(0)} = \beta^{(0)}$
			\end{tabular}
		\end{center}		
	\end{itemize}
	We want to show that if for every $\left[ \alpha' = \beta'\right] \in \Gamma$ we have $\alpha' \approx_{n} \beta'$, then we have $\alpha^{(0)} \approx_n \beta^{(0)}$ for all $n \geq 0$.
	
	We assume that we indeed have for every $\left[ \alpha' = \beta'\right] \in \Gamma$ we have $\alpha' \approx_{n} \beta'$.
	By induction hypothesis we know that if  $\D \alpha^{(k+1)} \approx_{n} \D \beta^{(k+1)}$ for some $k \geq 0$ and $n \geq 0$, then $\alpha^{(k)} \approx_{n} \beta^{(k)}$, which implies $\D \alpha^{(k)} \approx_{n+1} \D \beta^{(k)}$. Hence by chaining this property, we obtain that for any $k \geq 0$, $n \geq 0$: \[ \D \alpha^{(k+1)} \approx_{n} \D \beta^{(k+1)} \implies  \alpha^{(0)} \approx_{n+k} \beta^{(0)} \]
	Applying it with $n = 0$, and using the fact that $\approx_{0}$ is the total relation, we obtain for all $k \geq 0$
	\[ \alpha^{(0)} \approx_{k} \beta^{(0)} \]
\end{proof}

\subsection{Universality}

In this subsection, we prove the fullness of $\fa{-}$.

\begin{proposition}\label{app:prop:universal_proof}
	If $\C$ is pseudo-purifiable, then the functor $\fa{-} : \RegStSeq{\C} \to \Onion$ is full.
\end{proposition}
\begin{proof}
	We take $(f_k)_{k \geq 1} \in \Onion(\fa{A},\fa{B})$. By regularity, there is a $n$ and a $f^{\textup{ref}}$ and $f^{\textup{irreg}}$ such that for every $k \geq n$ we have 
	\begin{center}$f_k = \input{regularity.tikz}$
	\end{center}
	Since $\C$ is pseudo-purifiable, we can pseudo-purify its morphisms, so for every $k \geq 1$ we write $p(f_k)$ for an arbitrarily chosen pseudo-purification of $f_k$.
	
	We want to build $\alpha \in \RegStSeq{\C}(A,B)$ such that $\fa{\alpha}=(f_k)_{k \geq 1}$. We build $\alpha$ inductively, ensuring that at all rank $k < n$ we have that
	\[ p(f_k)= \input{open_onion_from_strat.tikz}  \]
	This ensures that we always have
	\[ f_k = \input{onion_from_strat_bis.tikz} \]	
	\begin{itemize}
		\item \emph{Initialisation:} we simply take $\alpha_1 = p(f_1)$.
		\item \emph{Irregular part:} We assume that for $k < n-1$ we have $\alpha_1,\dots,\alpha_k$ already defined and satisfying the hypothesis. By monotonicity, we know that we have
		\begin{center}
			$\input{monotone_succ.tikz} \preceq \input{monotone_pred.tikz}$
		\end{center}
		So using the definition of $\preceq$ we have $g \in \C(\fa{A}_{k+1},\fa{B}_k \otimes Y)$ and $h \in \C(Y,I)$ such that
		\begin{center}
			$\input{mono_succ.tikz} = \input{mono_dec_succ.tikz}$\\[5pt]
			$\input{mono_dec_pred.tikz} = \input{mono_pred.tikz}$
		\end{center}
		In those equation, we consider stateful morphisms $h \in \C(\text{In} \otimes M, \textup{Out} \otimes M')$ which we represent by having the initial state $M$ and final state $M'$ ``\textcolor{red}{vertical}'' while the standard input and output are ``\textcolor{blue}{horizontal}''. While this could be formalise in the context of a double category, we only use it here as a diagrammatic notation: we read diagrams from up/left to right/down.		
		
		We choose $p(g)$ and $p(h)$ two pseudo-purifications of respectively $g$ and $h$, and we obtain the following:
		\begin{center}
			$\input{pure_mono_succ.tikz} = \input{pure_dec_succ.tikz}$\\[5pt]
			$\input{pure_dec_pred.tikz} = \input{pure_mono_pred.tikz}$			
		\end{center}
		Using uniqueness of pseudo-purification up to a causal morphism, we obtain two causal morphisms $c$ and $d$ such that
		\begin{center}
			$\input{pure_mono_succ_bis.tikz} = \input{pure_dec_succ_bis.tikz}$\\[5pt]
			$\input{pure_dec_pred_bis.tikz} = \input{pure_mono_pred_bis.tikz}$			
		\end{center}
		It follows that 
		\begin{center}
			$\input{pure_mono_succ_bis.tikz} = \input{pure_recap.tikz}$
		\end{center}
		We then define $\alpha_{k+1} =  c \circ (\id \otimes p(h)) \circ d$. By construction, it satisfies the hypothesis.
		\item \emph{Regular part:} for $k \geq n$ we simply take $\alpha_k = f^{\textup{reg}}$.
	\end{itemize}
	By construction, we have $\fa{\alpha} = (f_k)_{k \geq 1}$, hence $\fa{-}$ is full.
\end{proof}

\subsection{Completeness}

In this subsection, we prove the completeness of $\fa{-}$. We start by a few lemmas.

\begin{lemma}\label{app:lem:shadow}
	In a category $\C$, if $\pi : X \to X$ is a shadow of a shadowful morphism $f \in \C(A,B \otimes X)$ then
	\[ (\id_B \otimes \pi) \circ f = f \]
\end{lemma}
\begin{proof}
	Using the definition of a shadowful morphism, we obtain
	\[ (\id_B \otimes \pi) \circ f = f \iff \pi \circ \pi = \pi \]
\end{proof}

\begin{lemma}\label{app:lem:subpurification}
	In a pseudo-purifiable category, if $f \in \C(A,B)$, $g \in \C(A,B \otimes X)$, $p \in \Pure(g)$, and
	\[f = \input{purification_G.tikz}\]
	then we have $p \in \Pure(f)$.
\end{lemma}
\begin{proof}
	Since $f$ is $g$ composed with the discard, then a pseudo-purification of $G$ composed with a pseudo-purification of $\Terre{X}$ gives a pseudo-purification of $f$. The identity morphism $\id_X$ is a pseudo-purification of $\Terre{X}$.
\end{proof}

\begin{lemma}\label{app:lem:complete}
	If $\C$ is a shadow pseudo-purifiable category, whenever $\fa{\alpha} = \fa{\beta}$ there exists $\alpha',\beta'$ such that $\fa{\alpha'} = \fa{\beta'}$ and 
	\[ \CM,\IM,[\D \alpha' = \D \beta'] \vdash \alpha = \beta \] 
\end{lemma}
\begin{proof}
	We take $\alpha,\beta \in \RegStSeq{\C}(A,B)$ such that $\fa{\alpha} = \fa{\beta}$. We take $p(\alpha_1)$ a shadowful pseudo-purification of $\alpha_1$, and $p(\beta_1)$ a shadowful pseudo-purification of $\beta_1$. Using Lemma \ref{app:lem:subpurification}, we note that $p(\alpha_1)$ is also a pseudo-purification of $f_1$, and so is $p(\beta_1)$. So using uniqueness of the pseudo-purification up to a causal morphisms, there exists $c$ causal such that:
	\begin{center}
		$\input{compl_a1.tikz} = \input{compl_pa1g.tikz} \qquad \input{compl_b1.tikz} = \input{compl_pa1Cg.tikz}$
	\end{center} 
	
	Since $\C$ is a shadow category, we write $\pi$ for the (idempotent) shadow of $p(\alpha_1)$, and from Lemma \ref{app:lem:shadow} we have
	\begin{center}$\input{compl_a1.tikz} = \input{compl_pa1PIg.tikz} \qquad \input{compl_b1.tikz} = \input{compl_pa1PICg.tikz}$
	\end{center}
	Using the fact that the discard is causal and $\pi$ is idempotent, we can rewrite $\alpha$ using $\IM$ and $\CM$ as follows:
	\begin{center}
		$\alpha = \input{compl_a_exp.tikz} \equiv \input{compl_a_eq.tikz} = (\id_{B_1} \otimes \D \alpha') \circ \gamma'$
	\end{center}
	where $\alpha'$ and $\gamma'$ are defined as follows:
	\begin{center}
		$\input{compl_c_prime.tikz} \df \input{compl_c_prime_def.tikz} \qquad \input{compl_a_prime.tikz} \df \input{compl_a_prime_def.tikz} $
	\end{center}
	Similarly, using the fact that $\C$ and the discard are causal and $\pi$ is idempotent, we can rewrite $\beta$ using $\IM$ and $\CM$ as follows:
	\begin{center}
		$\beta = \input{compl_b_exp.tikz} \equiv \input{compl_b_eq.tikz} = (\id_{B_1} \otimes \D \beta') \circ \gamma'$
	\end{center}
	where $\gamma'$ is defined above and $\beta'$ is defined as follows:
	\begin{center}
		$\input{compl_b_prime.tikz} \df \input{compl_b_prime_def.tikz}$
	\end{center}
	So we found $\alpha',\beta',\gamma'$ such that
	\[\CM,\IM \vdash \alpha = (\id_{B_1} \otimes \D \alpha') \circ \gamma'\]\[\CM,\IM \vdash \beta = (\id_{B_1} \otimes \D \beta') \circ \gamma' \]
	This means that
	\[ \CM,\IM, \left[\D \alpha' = \D \beta'\right] \vdash \alpha = \beta \]
	We still need to prove that $\fa{\alpha'} = \fa{\beta'}$. Since $\fa{\alpha} = \fa{\beta}$ and $\fa{-}$ is sound with respect to $\CM$ and $\IM$ (Lemma \ref{app:lem:soundness_proof}) we know that we have:
	\begin{center}
		$\fa{\input{compl_a_eq_bis.tikz}} = \fa{\input{compl_b_eq_bis.tikz}}$
	\end{center}
	Since $\pi$ is a shadow of $p(\alpha_1)$, then it is equivalent to:
	\begin{center}
		$\fa{\input{compl_a_prime_def.tikz}} = \fa{\input{compl_b_prime_def.tikz}}$
	\end{center}
	hence $\fa{\alpha'} = \fa{\beta'}$
\end{proof}

\begin{proposition}\label{app:prop:complete_proof}
	Whenever $\fa{\alpha} = \fa{\beta}$ we have $\CM,\IM \vdash \alpha = \beta$.
\end{proposition}
\begin{proof}
	We chain the use of lemma \ref{app:lem:complete} and obtain two sequences $\alpha^{(n)}$ and $\beta^{(n)}$ such that $\alpha^{(0)} = \alpha$, $\beta^{(0)}= \beta$ and for all $n \geq 0$ we have
	\[ \CM,\IM,[\D \alpha^{(n+1)} = \D \beta^{(n+1)}] \vdash \alpha^{(n)} = \beta^{(n)} \] 
	Using the coinduction rule this implies \[\CM,\IM \vdash \alpha = \beta\]
\end{proof}